\documentclass[aps,prl,reprint,showpacs,superscriptaddress,nofootinbib]{revtex4-2}
\usepackage[utf8]{inputenc}
\usepackage[T1]{fontenc}
\usepackage{amsmath, amsthm, amssymb, mathtools, dsfont} 
\usepackage{fixmath}
\usepackage{times}
\usepackage{graphicx}  
\usepackage{dcolumn} 
\usepackage{bm,bbm} 
\usepackage[normalem]{ulem}
\usepackage[usenames,dvipsnames]{xcolor}
\usepackage{esint}
\usepackage{paralist}
\usepackage[normalem]{ulem}

\usepackage{geometry}
\definecolor{myurlcolor}{rgb}{0,0,0.7}
\definecolor{myrefcolor}{rgb}{0.8,0,0}
\usepackage[unicode=true,pdfusetitle,  bookmarks=true,bookmarksnumbered=false,bookmarksopen=false, breaklinks=false,pdfborder={0 0 0},backref=false,colorlinks=true, linkcolor=myrefcolor,citecolor=myurlcolor,urlcolor=myurlcolor]{hyperref}
\DeclareGraphicsExtensions{.png,.pdf,.eps}
\graphicspath{ {./figures/} }
\usepackage{epstopdf}

\usepackage{mathtools}

\makeatletter

 \theoremstyle{plain}
 
 \theoremstyle{plain}
 \newtheorem{lem}{Lemma}
 \theoremstyle{plain}
 \newtheorem{thm}{Theorem}
 \theoremstyle{plain}
  \newtheorem{prop}{Proposition}
 \theoremstyle{plain}
 
 \theoremstyle{plain}
 
 \theoremstyle{plain}
 
 \theoremstyle{remark}
 \newtheorem*{rem*}{Remark}
 \theoremstyle{plain}
  \newtheorem{rem}{Remark}
\newtheorem{definition}{Definition}
\theoremstyle{plain}
 \newtheorem*{conj*}{Conjecture}
 \theoremstyle{plain}

\renewcommand{\sout}[1]{}
\newcommand{\er}[1]{} 
\newcommand{\m}{m} 
 
\newcommand{\e}{\mathrm{e}}
\newcommand{\ot}{\otimes}
\newcommand{\unity}{\mathbbmss{1}}

\renewcommand{\exp}{\mathrm{exp}}
\DeclareMathOperator{\tr}{tr}

\renewcommand{\H}{\mathcal{H}}

\newcommand{\Complex}{\mathbb{C}} 
\newcommand{\R}{\mathbb{R}} 
\newcommand{\M}{\mathbf{M}} 
\newcommand{\N}{\mathbf{N}} 
\renewcommand{\L}{\mathbf{L}} 
 
\newcommand{\Q}{\mathcal{Q}} 

\newcommand{\nop}{24}
\newcommand{\noq}{1299}

\renewcommand{\P}{\mathcal{P}} 
\newcommand{\PPP}{\mathbf{P}} 

\newcommand{\pr}{ \mathbb{P}} 

\newcommand{\q}{\mathrm{q}_\mathrm{succ}}

\newcommand{\Unitary}{\mathrm{U}} 

\newcommand{\norm}[1]{\left|\left|  \ #1  \ \right| \right|}
\newcommand{\HS}[1]{\left| \left| \ #1 \ \right| \right|_2}

\newcommand{\NN}[1]{\N_{{X}_{#1}}}

\newcommand{\tv}{\mathrm{d}_{\mathrm{TV}}}

\newcommand{\HE}[1]{\left\langle #1 \right\rangle_{\mathrm{Haar}}}

\newcommand{\Sc}[1]{S^{#1}_{\Complex}} 
\newcommand{\Sr}[1]{S^{#1}} 
\newcommand{\Sx}{S_{\scriptscriptstyle{X}}} 
\newcommand{\Ej}{E_{\scriptscriptstyle{X}}} 
  
\newcommand{\Unitaryx}{U_{\scriptscriptstyle{X}}}
\newcommand{\Wx}{W_{\scriptscriptstyle{X}}}

\newcommand{\intH}{{\int} d\mu_\bo(U)}

\newcommand{\all}[2]{w_{\scriptstyle #1 }^{\scriptscriptstyle #2}}

\newcommand{\so}{m}
\newcommand{\bo}{n}

\newcommand{\Sm}{\mathbb{S}_{\so}}

\renewcommand{\S}{\mathbb{S}}

\newcommand{\Rb}{R^{(\so)}\left(\M\right)}     

\newcommand{\np}[1]{ \left| \left| \ P \ket{ #1} \ \right| \right|}
\newcommand{\nup}[1]{ \left| \left| \ PU \ket{#1} \ \right| \right|}

\newcommand{\psucc}[1]{\mathrm{P}_{\mathrm{succ}} \left( \mathcal{E}, #1 \right)}

\newcommand{\sk}[1]{\scriptstyle | #1 \rangle}

\newcommand{\rbracket}[1]{\left(#1\right)} 
\newcommand{\sbracket}[1]{\left[#1\right]} 
\newcommand{\cbracket}[1]{\left\{#1\right\}} 

\newcommand{\p}[1]{\mathrm{\mathbf{#1}}}

\newcommand{\pv}[2]{\mathbf{p}\rbracket{#1|#2}}

\definecolor{dukeblue}{rgb}{0.0, 0.0, 0.61}
\definecolor{cadmiumgreen}{rgb}{0.0, 0.42, 0.24}

\global\long\global\long\global\long\def\bra#1{\mbox{\ensuremath{\langle#1|}}}
\global\long\global\long\global\long\def\ket#1{\mbox{\ensuremath{|#1\rangle}}}
\global\long\global\long\global\long\def\bk#1#2{\mbox{\ensuremath{\ensuremath{\langle#1|#2\rangle}}}}
\global\long\global\long\global\long\def\kb#1#2{\mbox{\ensuremath{\ensuremath{\ensuremath{|#1\rangle\!\langle#2|}}}}}

\makeatother
\renewcommand{\ket}[1]{\left| #1 \right>} 
\renewcommand{\bra}[1]{\left< #1 \right|} 
\newcommand{\braket}[2]{\langle #1 | #2 \rangle} 
\newcommand{\ketbra}[2]{\left| #1 \rangle\langle #2 \right|} 

\begin{document} 	
\title{Implementation of quantum measurements using classical resources and only a single ancillary qubit}
	 
		\author{Tanmay Singal}
	\email{tanmaysingal@gmail.com}
	\affiliation{Department of Analysis, Budapest University of Technology and Economics,
1111 Budapest, Egry József u. 1., Hungary}
\affiliation{Center for Theoretical Physics, Polish Academy of Sciences, Al. Lotnik\'ow 32/46, 02-668 Warsaw, Poland}

	\author{Filip B. Maciejewski}
	\email{filip.b.maciejewski@gmail.com}
	\affiliation{Center for Theoretical Physics, Polish Academy of Sciences, Al. Lotnik\'ow 32/46, 02-668 Warsaw, Poland}

	\author{Micha\l\ Oszmaniec}
	\email{oszmaniec@cft.edu.pl}
	\affiliation{Center for Theoretical Physics, Polish Academy of Sciences, Al. Lotnik\'ow 32/46, 02-668 Warsaw, Poland}

\begin{abstract}
We propose a scheme to implement general quantum measurements, also known as Positive Operator Valued Measures (POVMs) in dimension $d$ using only classical resources and a single ancillary qubit. Our method is based on probabilistic implementation of $d$-outcome measurements which is followed by postselection of some of the received outcomes. 
We conjecture that success probability of our scheme is larger than a constant independent of $d$ for all POVMs in dimension $d$. 
Crucially, this conjecture implies the possibility of realizing arbitrary nonadaptive quantum measurement protocol on $d$-dimensional system using a single auxiliary qubit with only a \emph{constant} overhead in sampling complexity.
We show that the conjecture holds for typical rank-one Haar-random POVMs in arbitrary dimensions. Furthermore, we carry out extensive numerical computations showing success probability above a constant for a variety of extremal POVMs, including SIC-POVMs in dimension up to $\noq$.
Finally, we argue that our scheme can be favorable for experimental realization of POVMs, as noise compounding in circuits required by our scheme is typically substantially lower than in the standard scheme that directly uses Naimark’s dilation theorem.

\end{abstract}
\maketitle	
 
Quantum measurements recover classical information stored in quantum systems and, as such, constitute an essential part of virtually any quantum information protocol. Every physical platform has its native measurements that can be realized with relative ease. In many cases, the class of easily implementable measurements contains projective (von Neumann) measurements. However, there are numerous applications  \cite{Knill1997, Briegel2009,Gisin2007, Bergou2010,Braunstein1994,Braunstein1996,Toth2014,Pirandola2018, Degen2017} in which more general quantum measurements, so called Positive-Operator-Valued Measures (POVMs), need to be implemented. 
Implementation of these measurements requires additional resources. A recent generalization  \cite{Oszmaniec17} of  Naimark's dilation theorem \cite{Peres2002} showed that the  most general measurement on $N$ qubits requires $N$ auxiliary qubits, when projective measurements can be implemented on the combined system in a randomized manner.

From the perspective of implementation in near-term quantum devices \cite{Preskill2018}, it is desirable to implement arbitrary POVMs with fewer resources.
Particularly, one would like to reduce the number of auxiliary qubits needed to implement a complex quantum measurement. 
A related problem is to quantify the relative power that generalized measurements in $d$-dimensional quantum systems have with respect to projective measurements in the same dimension. 
While POVMs appear as natural measurements for a variety of quantum information tasks: quantum state discrimination \cite{QuantumStateDisriminationRev}, quantum tomography \cite{Derka1998,Renes2004,OptTomography}, multi-parameter metrology \cite{Ragy2016,Szcyykulska2016}, randomness generation \cite{Acin2016}, entanglement \cite{Shang2018}  and nonlocality detection \cite{Vertesi2010},  hidden subgroup problem \cite{Bacon2006,HSP}, port-based-teleportation \cite{Ishizaka2008,Studzinski2017port,Mozrzymas2018optimal}, to name just a few. 
It is, however, not clear in general what quantitative advantage the more complex measurements offer over their simpler projective counterparts. 
This is because of the possibility to realize non-projective quantum measurements via randomization and post-processing of simpler measurements \cite{Davies1976, Giulio2004, Buscemi2005, DAriano2005, Ali2009,Oszmaniec17,Oszmaniec19}. 
Specifically, taking convex combinations of projective measurements can result in implementation of a priori quite complicated nonprojective POVMs \cite{Oszmaniec17,Oszmaniec19}.

In this work we advance understanding of the relative power between projective and generalized measurements by focusing on a simpler problem, namely the relation between $d$-outcome POVMs and general (with arbitrary number of outcomes) POVMs acting on a $d$-dimensional Hilbert space $\H\approx\Complex^d$.  
We find a strong evidence that general quantum measurements do not offer an asymptotically increasing advantage over $d$-outcome POVMs for  general quantum state discrimination problems \cite{QuantumStateDisriminationRev}, as $d$ tends to infinity.
Specifically, we generalize the method of POVM simulation from \cite{Oszmaniec19} based on randomized implementation of  restricted-class POVMs, followed by post-processing and postselection (defined later, see also Fig. \ref{fig:general_figure}).  
Here by postselection we mean disregarding certain measurement outcomes and accepting only the selected ones. 
In \cite{Oszmaniec19} it was shown that postselection allows to implement arbitrary POVM on $\Complex^d$ using only projective measurements and classical resources.
This, however, comes with a cost - the method outputs a sample from a target quantum measurement with success probability $\q=\frac{1}{d}$. 
In this work we find that, surprisingly, there exists a protocol that allows to simulate a very broad class of POVMs on $\Complex^d$ via $d$-outcome POVMs and postselection with success probability $\q$ above a constant which is independent on the dimension $d$. 
Importantly, our construction ensures $d$-outcome POVMs used in the simulation can be implemented using projective measurements in Hilbert space of dimension $2d$.
Therefore, our method gives a way to implement quantum measurements on $\Complex^d$ using only a single auxiliary qubit and projective measurements with constant success probability. 
We note that there exist schemes implementing arbitrary POVMs on $\Complex^d$ using a sequence of von Neumann instruments (i.e., a description of quantum measurements which includes post-measurement state of the system) on a system extended by a single auxiliary qubit \cite{Anderson2008,Bouda2020}. 
Our method is potentially simpler to implement as, in a given round of the experiment, only a single projective measurement has to be realized on the extended system and post-measurement states need not to be considered.

\begin{figure*}[t]
  \includegraphics[width=\linewidth]{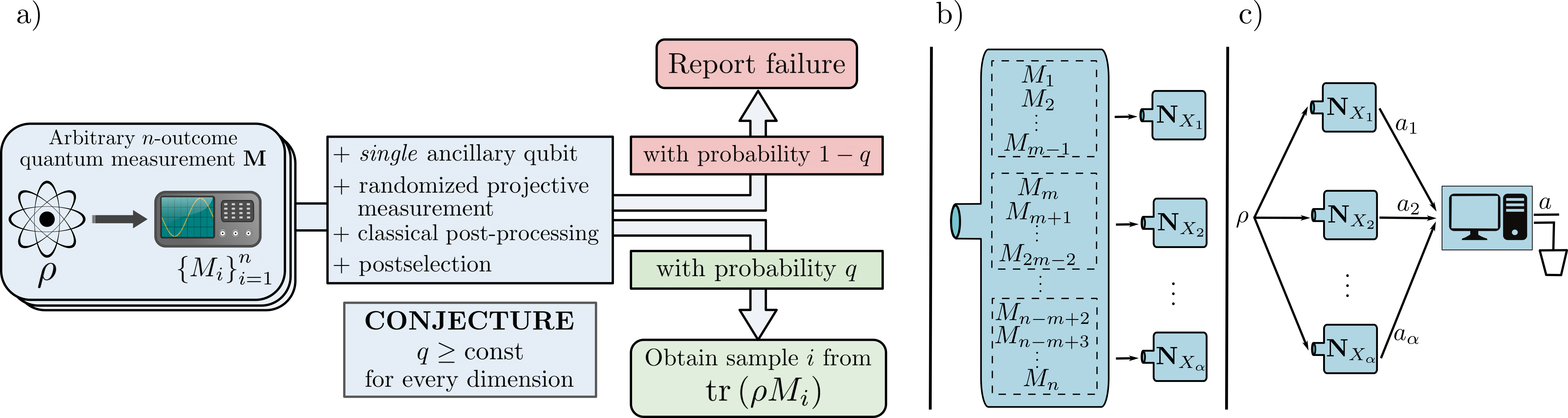}
    \caption{
    Implementation of a target measurement $\M$ with POVMs having at most $\so$ outcomes and postselection.
    Left figure illustrates a) general idea of the scheme, while in the right figure the method is illustrated in more detail -- in
    b), the $\so$-outcome POVMs $\NN{j}$ are constructed using effects of $\M$ that correspond to different subsets $X_\gamma$ forming a partition of $[\bo]$ into subsets of cardinality $\so-1$ (figure shows the standard partition and effects of $[\bo]$: $X_1 = \left\{ 1,2,\cdots, \so-1\right\}$, $X_2 = \left\{ \so, \cdots, 2 \so -2  \right\}$, etc. )  In c), POVMs $\N^{X_\gamma}$ are  implemented probabilistically and the resulting outcomes $a_i$ undergo suitable post-processing and post-selection steps which simulate $\M$.
    }
    \label{fig:general_figure}
\end{figure*}

While we do not prove that the success probability $\q$ of our scheme is lower bounded by a dimension-independent constant for any POVMs on $\Complex^d$, we give strong evidence that this is indeed the case.
First, we prove that for generic $d$-outcome Haar-random rank-one POVMs in $\Complex^d$ \cite{HJMN20} the success probability is above $6.5\%$ (numerically we observe $\approx 25\%$).  
We also support our conjecture by numerically studying specific examples of symmetric informationally complete POVMs (SIC-POVMs) \cite{Scott2010symmetric, Appleby2005, Fuchs2017SIC} and for a class of nonsymmetric  informationally complete POVMs \cite{DAriano2004} (IC-POVMs), both for dimensions up to $\noq$. 
As the dimension increases, we observe that the success probability $\q$ both for SIC-POVMs and IC-POVMs is $\approx 1/5$.
Importantly, if true, our conjecture implies that \emph{any} non-adaptive measurement protocol can be realized using only single ancillary qubit with a sampling overhead that does not depend on the system size.

Finally, our scheme gives a possibility of more reliable implementation of complicated POVMs in noisy quantum devices. To support this claim, we employ the noise model used in Google's recent demonstration of quantum computational advantage \cite{Google2019}. We make the following comparison between our method and the standard Naimark's scheme of POVM implementation: for implementing  typical random POVMs on $N$ qubits, the fidelity of circuits which implement our scheme is exponentially higher than for Naimark's implementation. This is due to the lower number of ancillary qubits required.
 
\textit{Preliminaries---} 
We start by introducing notation and the concepts necessary to explain our POVM implementation scheme. 
We will be studying generalized quantum measurements on $d$-dimensional Hilbert space $\H\approx \Complex^d$.  
An $\bo$-outcome POVM, is an $\bo$-tuple of linear operators on $\Complex^d$ (usually called effects), i.e., $\M  = \left(M_1,M_2,\cdots,M_\bo \right)$, satisfying  $M_i \ge 0$ and  $\sum_{i=1}^\bo M_i = \unity$, where $\unity$ is identity on $\Complex^d$. A POVM $\mathbf{P}=\left(P_1,P_2,\cdots,P_n\right)$ is called projective if all its effects satisfy the following relations: $P_i P_j = \delta_{ij} P_i$. Measurement of $\M$ on a quantum state $\rho$ results in a random outcome $i$, distributed according to the Born rule $\mathrm{p}(i|\rho, \M)= \tr\left( \rho M_i \right)$. We will denote the set of all all $\bo$-outcome POVMs by $\P(d,\bo)$. The set $\P(d,\bo)$ is convex \cite{DAriano2005}: for $\M, \N \in \P(d,\bo)$, and  $p\in[0,1]$  we define $p \M + (1-p) \N$ to be an $\bo$-outcome POVM with the $i$-th effect given by $\left[ p \M + (1-p) \N \right]_i=  p M_i + (1-p) N_i$. A convex mixture  $p \M + (1-p) \N$  can be operationally interpreted as a POVM realized by applying, in a given experimental run, measurements $\M,\N$ with probabilities $p$ and $1-p$ respectively. A POVM $\M\in\P(d,n)$ is called extremal if it cannot be decomposed as a nontrivial convex combination of other POVMs.

Another classical operation that can be applied to POVMs is classical post-processing \cite{Buscemi2005,Haapasalo2012}: given a POVM $\M$, we obtain another POVM $\Q\left( \M \right)$ by probabilistically relabeling the outcomes of the measurement $\M$. Effects of $\Q\left( \M \right)$ are given by $\Q(\M)_i = \sum_{j} q_{i|j} M_j$, where $q_{i|j}$ are conditional probabilities, i.e.,  $q_{i|j}\ge0$ and $\sum_{i} q_{i|j} =1$. Lastly, postselection, i.e., the process of disregarding certain outcomes can be used to  implement otherwise inaccessible POVMs. 
We say that a POVM $\L=(L_1,\ldots,L_\bo,L_{\bo+1})$ simulates a POVM $\M=(M_1,\ldots,M_\bo)$ with postselection probability $q$ if $L_i=q M_i$ for $i=1,\ldots,\bo$. 
This nomenclature is motivated by realizing that when we implement $\L$, then, conditioned on getting the first $\bo$ outcomes, we obtain samples from $\M$.  Thus, we can simulate $\M$ by implementing $\L$, and post-selecting on non-observing outcome $\bo+1$.
The probability of successfully doing so is $q$ which means that a single sample of $\M$ is obtained by implementing $\L$  on average $1/q$ number of times. The reader is referred to \cite{Oszmaniec19} for a more detailed discussion of simulation via post-selection.

We will use $\|A\|$ to denote the operator norm of a linear operator $A$, and $[\bo]$ to denote $\bo$-element set $\lbrace1,\ldots\bo\rbrace$. Moreover,  we will use  $\mu_\bo$ to refer to Haar measure on $\bo$-dimensional unitary group $\Unitary(\bo)$, and by $\pr_{U\sim\mu_\bo}(\mathcal{A})$ we will denote  probability of occurrence of an event $\mathcal{A}$ according to this probability measure. Finally, for two positive-valued functions $f(x),g(x)$ we will write $f=\Theta(g)$ if there exist positive constants $c,C>0$ such that $c f(x)<g(x)<C f(x)$, for sufficiently large $x$.

\textit{General POVM simulation protocol---} The following theorem gives a general lower bound on the success probability of simulation of $\bo$-outcome POVMs via measurements with bounded number of outcomes and postselection.  

\newpage

\begin{thm}
 \label{thm:scheme}
Let $\M=(M_1,M_2,\ldots,M_{\bo})$ be an $n$-outcome POVM on $\Complex^d$. Let $\so \leq d$ be a natural number and let $\{X_{\gamma}\}_{\gamma=1}^{\alpha}$ be a partition of $[\bo]$ into  disjoint subsets $X_{\gamma}$ satisfying $|X_{\gamma}|\leq \so-1$. Then, there exists a simulation scheme that uses measurements having at most $\so$ outcomes, classical randomness and post-selection that implements $\M$ with  success probability
    \begin{align}
    \label{eq:qsucc0} 
 \q = \left( \sum_{\gamma=1}^{\alpha} \left\| \sum_{i \in X_\gamma} M_{i} \right\| \right)^{-1}\ .
\end{align}
Furthermore, if $\mathrm{rank}M_i\leq1$, and $m\leq d$,  then measurements realizing the scheme can be implemented by projective measurements in dimension $2d$, i.e., using a single auxiliary qubit.
\end{thm}

\begin{proof} 
In what follows we give an explicit simulation protocol that generalizes  earlier result from \cite{Oszmaniec19,HiddenNonlocality2013} that  concerned the case of simulation via dichotomic measurements ($\so=2$). The idea of the scheme is given in Fig.~\ref{fig:general_figure}.  We start by defining, for every element $X_{\gamma}$ of the partition, auxiliary measurements $\N^{X_{\gamma}}$, each having $\so + 1$ outcomes, whose purpose is to "mimick" measurement $\M$ for outputs belonging to $X_{\gamma}$ and collect other (i.e., not belonging to $X_{\gamma}$) results in the "trash" output labelled by $\bo + 1$. 
Effects of $\N^{X_{\gamma}}$ are defined by $N^{X_{\gamma}}_i = \lambda_\gamma M_i$ for  $i\in X_{\gamma}$, $N^{X_{\gamma}}_i=0$ for $i\in[\bo]\setminus X_\gamma$, and $N^{X_{\gamma}}_{n+1}=\unity -  \lambda_\gamma \sum_{i\in X_\gamma}M_i$, where $\lambda_\gamma=\|\sum_{i\in X_{\gamma}} M_i\|^{-1}$. 

We then define a probability distribution $\lbrace\frac{\q}{\lambda_\gamma}\rbrace_{\gamma=1}^\alpha$. The simulation of $\M$ is realized by considering a convex combination of $\N^{X_{\gamma}}$ according to this distribution: $\L=\sum_{\gamma=1}^\alpha \frac{\q}{\lambda_\gamma} \N^{X_{\gamma}}$. An explicit computation shows that we have $L_i =\q M_i$, for $i\in[n]$ and therefore $\L$ simulates the target measurement $\M$ with success probability $\q$.  

Finally, each of the measurements $\N^{X_{\gamma}}$ comprising $\L$ has at most $|X_{\gamma}|+1$ nonzero effects and therefore they can be implemented with POVMs with at most $\so$ outcomes. From the standard Naimark scheme of implementation of POVMs (c.f. \cite{Peres2002}) we see that the dimension needed to implement a POVM $\N^{X_{\gamma}}$ via projective measurements equals at most the sum of ranks of effects of $\N^{X_{\gamma}}$.
In the case of rank-one $\M$  and $|X_\gamma|\leq\so -1 $ this sum for each $\N^{X_{\gamma}}$ is at most $d + \so -1 \leq 2d$, which completes the proof.
\end{proof}

Crucially, we recall that an arbitrary quantum measurement on $\Complex^d$ can be implemented by a convex combination of rank-one POVMs having at most $d^2$ outcomes followed by suitable post-processing \cite{Davies1976,DAriano2005}. This implies that our protocol facilitates the simulation of \textit{any} POVM on $\Complex^d$ using only a single ancillary qubit -- first by decomposing the target POVM into a convex combination of rank-one $\leq d^2$-outcome measurements, and subsequently applying Theorem~\ref{thm:scheme} to each of them. 

Importantly, the standard Naimark's implementation of a general POVM would require appending an extra system of dimension $d$ (which can be realised by $\log_2 d$ ancillary qubits) and carrying out a global projective measurement. Our simulation protocol greatly reduces this requirement on the dimension cost of implementing $\M$ with the possible  downside being the probabilistic nature of the scheme. 
The success probability $\q$ depends on the choice of the partition $\lbrace X_\gamma\rbrace_{\gamma=1}^\alpha$, and finding the optimal one (for a given bound on the size of $X_\gamma$)  is in general a difficult combinatorial problem. In what follows we collect analytical and numerical results suggesting the following

\begin{conj*}
For arbitrary extremal rank-one POVM $\M=(M_1,\ldots,M_\bo)$ on $\Complex^d$, there exists a partition $\lbrace X_\gamma\rbrace_{\gamma=1}^\alpha$ of $[\bo]$ satisfying $|X_\gamma|\leq d-1$  such that  the corresponding value of success probability $\q$ from Eq. \eqref{eq:qsucc0} is larger than a positive constant independent of $d$.
\end{conj*}  
Let us explore the intriguing conceptual consequences of the validity of this conjecture.
First, consider a general nonadaptive measurement protocol that utilizes some quantum measurement $\M$ on $\Complex^d$. 
Such a protocol consists of $S$ independent measurement rounds of a quantum state $\rho$ resulting in outcomes $i_1,i_2,\ldots,i_S$ distributed according to the probability distribution $p(i|\M,\rho)=\tr(M_i \rho)$. 
This experimental data is then processed to solve a specific problem at hand. 
If we can simulate any \emph{arbitrary} $\M$ (see comment below proof of Theorem~\ref{thm:scheme}) via POVMs that can be implemented using only a single auxiliary qubit with probability $q$,which is independent of the dimension  $d$, then this means that we can, on average,  exactly  reproduce the implementation of the above protocol for $q S$ of the total $S$ rounds. Importantly, we also know in which rounds the simulated protocol was successful, so we know which part of the output data generated by our simulation comes from the target distribution. Crucially, the above considerations are completely oblivious to the figure of merit and the structure of the problem that measurements of $\M$ aim to solve. 

For many quantum information tasks, losing only a constant fraction of the measurement rounds is not prohibitive and hence, assuming the validity of the conjecture, our POVM simulation scheme offers a way to significantly reduce quantum resources needed for said POVM's implementation. Such exemplary  tasks include quantum state tomography \cite{OptTomography}, quantum state discrimination \cite{QuantumStateDisriminationRev}, multi-parameter quantum metrology \cite{Ragy2016,Szcyykulska2016}  or port-based teleportation \cite{Ishizaka2008}, and will be explored by us in future works.

Our simulation protocol and the above conjecture are also relevant from the perspective of POVM simulability \cite{Oszmaniec17,Guerini2017, Oszmaniec19} that attracted a lot of attention recently in the context of resource theories \cite{OszmaniecBiswas2019,Uola2019, Carmeli2019, Takagi2019, Skrzypczyk2019,kuramochi2020compact,ResTheoryMeas}.
Namely, the maximal post-selection probability, $q^{(\so)}(\M)$, with which a target POVM $\M$ on $\Complex^d$ can be simulated using  strategies utilizing randomized POVMs with at most $\so$ outcomes, quantifies how far $\M$ is from the set of $\so$-outcome simulable POVMs in $\Complex^d$, denoted by $\S_\so$. 
Moreover, $q^{(\so)}(\M)$ imposes  bounds on the so-called white noise critical visibility $t^{(\so)}(\M)$ \cite{Oszmaniec17} and the robustness $R^{(\so)}(\M)$  \cite{OszmaniecBiswas2019} against simulation via POVMs from $\S_\so$. Here by critical visibility we mean a parameter $t^{(\so)}\left(\M\right)$ associated with a minimal amount of white noise that ensures that noisy version of $\M$ belongs to subset $\Sm$, namely
\begin{equation}
 t^{(\so)}\left(\M\right) \coloneqq \; \max \; \left\{ \; t \; | \; \Phi_t\left( \M \right) \in \Sm \right\},\\
\end{equation}
where $\Phi_t\left( \M \right)$ is a POVM with effects $\Phi_{t}(M_i):=t M_i + (1-t)\frac{\tr M_i}{d} \unity$. 
By robustness $R^{(\so)}(\M)$ with respect to $\Sm$, we mean the minimal amount of mixing of $\M$ with a POVM from $\Sm$ so that the resulting POVM belongs to $\Sm$, i.e.,
\begin{equation}
 R^{(\so)}\left(\M\right) \coloneqq \; \min \; \left\{ \; s \; | \; \exists \; \mathbf{K} \; \mathrm{s.t.} \; \frac{\M + s \mathbf{K}}{1+s} \in \Sm \right\}.
\end{equation}
Now, the above quantities are bounded with the success probability of our scheme via (see Appendix \ref{app:rela_q_R_V}): 

\begin{equation}
\label{eq:rela_q_R_V}
    q^{(\so)}(\M) \leq t^{(\so)}(\M)\ ,\ R^{(\so)}(\M) \leq \frac{1}{q^{(\so)}(\M)} -1 \ .
\end{equation}
Importantly, we note that the robustness $R^{(\so)}(\M)$ has an appealing operational interpretation: it is also expressible as the maximal relative advantage that $\M$ offers over any POVM in $\Sm$ for a state discrimination task \cite{OszmaniecBiswas2019}:
\begin{equation} 
 \label{eq:Rb_MED_gain_main}
 R^{(\so)}(\M) \ = \max_{\mathcal{E}} \dfrac{\mathrm{P}_{\mathrm{succ}} \left( \mathcal{E},\M \right)}{\underset{ \; \; \N \in \Sm}{\max}\,\mathrm{P}_{\mathrm{succ}} \left( \mathcal{E},\N \right)} \; \; - \; \;1,
\end{equation} 
where $\mathcal{E} = \left\{ (q_i, \sigma_i) \right\}_{i=1}^{\bo}$ is an ensemble of quantum states, and $\psucc{\M}$ is the probability for the minimum error discrimination of the states from $\mathcal{E}$ with $\M$. 
Now, from the second inequality in \eqref{eq:rela_q_R_V} and  the (conjectured) constant lower bound on $q^{(d)}$ we get a surprising conclusion: general POVMs on $\Complex^d$ \emph{do not offer} asymptotically increasing (with $d$) advantage over $d$-outcome simulable measurements for general quantum state discrimination problems.

\textit{Haar Random POVMs---} 
We \sout{would} want to qualitatively understand how $\q$ depends on the total number of outcomes $\bo$, the number of POVM outcomes used in the simulation $\so$, and the dimension $d$. 
To make the problem feasible\sout{,} we turn to study Haar-random POVMs on $\Complex^d$. Quantum measurements comprising this ensemble can be realized by a construction motivated by Naimark's extension theorem: (i) attach to $\Complex^d$ an ancillary system $\Complex^a$  so that  the composite system is  $\bo$-dimensional: $\Complex^d \ot \Complex^a \approx \Complex^\bo$, (ii) apply on this composite system a random unitary $U$ chosen from the  Haar measure $\mu_\bo$ in $\Unitary(\Complex^\bo)$, and (iii) measure the composite system in the computational basis. Effects of this measurement $\M^U$ are given by $M^U_i = \tr_{\Complex^a}\left( \unity \ot \kb{0}{0} \ U^\dag \kb{i}{i} U\right)$, where $\ketbra{0}{0}$ is a fixed state on $\Complex^a$. Haar-random POVMs were introduced first in \cite{HSP} in the context of the hidden subgroup problem and are a special case of a more general family of random POVMs studied recently in \cite{HJMN20}. Measurements $\M^U$ are extremal for almost all $U\in\Unitary(\bo)$. Furthermore, all extremal rank-one POVMs in $\Complex^d$ are of the form $\M^U$ for some $U\in\Unitary(\bo)$, and $\bo\in\lbrace d,d+1,\ldots,d^2\rbrace$. Hence, Haar-random POVMs form an ensemble consisting of extremal non-projective measurements, making them a natural test-bed for studying the performance of  our simulation algorithm.


\begin{thm}[Success probability of the implementation of Haar-random POVMs]
\label{thm:conc_q_lower_upper}
Let $\bo \in \left\{ d,\ldots,d^2 \right\} $, $\so\leq d$. Let $\M^U$ denote a rank-one $\bo$-outcome Haar-random POVM on $\Complex^d$. Let $\q^{(\so)}(\M^U)$ denote the success probability of implementing $\M^U$ via $\so$-outcome measurements as in Eq. \eqref{eq:qsucc0} for the standard partition of $\left[n\right]$, i.e., $X_1=\left\{ 1,\ldots  \so-1 \right\},\ X_2=\left\{ \so,\so+1,\ldots, 2\so-2 \right\}$, etc.
We then have
\begin{equation}
 \label{eq:conc_q_lower_simpler1}
  \underset{U\sim \mu_\bo}{\pr}\left( \q^{(\so)}(\M^U) \geq \Theta(\frac{m}{d})  \right)  \rightarrow 1 ,\text{as } d\rightarrow \infty\ .
\end{equation}
Moreover, let $q^{(\so)}(\M^U)$ be the maximal success probability of implementing $\M^U$ with postselection via convex combination of $\so$-outcome measurements using any simulation protocol. 
We then have 
\begin{equation}
 \label{eq:conc_q_upper_simpler1}
  \underset{U\sim \mu_\bo}{\pr}\left( q^{(\so)}(\M^U) \leq \Theta(\frac{m}{d}\log(d))  \right)  \rightarrow 1 ,\text{as } d\rightarrow \infty\ .
\end{equation}
\end{thm}
The above result shows that when simulating Haar-random POVMs on $\Complex^d$ with $\so$-outcome measurements in our scheme, the success probability scales as $\frac{\so}{d}$. 
Furthermore, Eq. \eqref{eq:conc_q_upper_simpler1} shows the optimality of our method up to a factor logarithmic in $d$. Specifically,  we obtain the following crucial result: when $\so=d$, with overwhelming probability over the choice of random $U\in\Unitary(d^2)$, \sout{the} $\q^{(d)}(\M^U)$ is above $6.74\%$. Below we sketch the proof for Theorem \ref{thm:conc_q_lower_upper}. We provide a complete proof in Appendix~\ref{sec:conc_q_lower}, \sout{together with finite $d$} with expressions for finite $d$, for bounds in Eq. \eqref{eq:conc_q_lower_simpler1} and \eqref{eq:conc_q_upper_simpler1}.

\begin{proof}[Sketch of Proof]

An explicit computation shows that for any subset $X\subset [\bo]$, we have $\| \sum_{i\in X} M^U_i \| = \| U_{X} \|^2$, where $U_X$ is a $d\times |X|$ matrix, obtained by choosing the first $d$ rows of $U$, and then taking from the resulting matrix those columns with indices in $X$. With this we analyze the statistical behaviour of $\q(\M^U)$ in the regime  $d\rightarrow \infty$ using tools from random matrix theory.  Specifically, the proof relies on the phenomenon of concentration of measure \cite{Aubrun2017} on the unitary group $\Unitary(\bo)$ equipped with the Haar measure and distance induced by the Hilbert-Schmidt norm. It shows that as $\bo\longrightarrow \infty$, Lipschitz-continuous random variables on  $U(\bo)$ are with high probability close to their Haar-averages - this is captured by large deviation bounds (also known as concentration inequalities), that upper bound the probability that a random variable take values drastically different form its Haar-average.

In order to prove Eq.~\eqref{eq:conc_q_lower_simpler1}, we choose $\| U_{X} \|$ as the random variable to which we apply the machinery of concentration of
measure. An upper bound to its Haar-average is obtained by performing a discrete optimization over an $\epsilon$-net of an $m-1$-dimensional complex sphere. Since the concentration inequality is true for all subsets $X$ in the partition of $[\bo]$, the union bound shows that $\sum_X\| \sum_{i\in X} M^U_i \|$ also exhibits concentration of measure, which gives Eq. \eqref{eq:conc_q_lower_simpler1}.

In order to prove Eq. \eqref{eq:conc_q_upper_simpler1}, we invoke the inequality in Eq.~\eqref{eq:rela_q_R_V}, and use it to upper bound $q^{(\so)}$ \sout{via} with the robustness $R^{(\so)}(\M^{U})$ of a random POVM $\M^{U}$ \sout{w.r.p.} with respect to  $\so$-outcome simulable POVMs in $\Complex^{d}$. 
Using the interpretation of robustness in the context of state-discrimination (see Eq.~\eqref{eq:Rb_MED_gain_main}), we lower bound it by constructing a specific ensemble of quantum states obtained by rescaling the effects of $\M^{U}$.
In this way, a lower bound on the robustness (hence an upper bound on the success probability) becomes a function of the matrix elements $|U_{ij}|^2$ of the Haar-random unitary $U$. Finally, we prove a concentration of measures inequality for this resulting function, by again invoking the union bound and the  cumulative distribution function of $|U_{ij}|^2$, which was obtained in \cite{Zyczkowski2000}. 
\end{proof}

\textit{Numerical results---}
We tested the performance of our POVM simulation scheme by computing $\q$ for SIC-POVMs \cite{Scott2010symmetric, Appleby2005, Fuchs2017SIC}, IC-POVMs \cite{DAriano2004} and for Haar-random $d^2$-outcome POVMs. 
We focused on simulation strategies via POVMs that can be implemented with a single auxiliary qubit (this corresponds to setting $\so=d$ in Theorem \ref{thm:scheme}). For every dimension, we generated effects of symmetric  POVMs numerically from a single  fiducial pure state via transformations $X_d^i Z_d^j$ , where $i,j\in [0,d-1]$ and $X_d,Z_d$ are $d-$dimensional analogues of Pauli $X$ and $Z$ operators. 
For IC-POVMs we used a one-parameter family of fiducial states $\ket{\psi_\alpha}$ described in Ref~\cite{DAriano2004} for the  specific value $   \alpha = \frac{1}{2}\left(1+i\right)$ (we remark that POVMs originating from other values of $\alpha$ exhibited a similar behaviour). For SIC-POVMs we used fiducial states from a catalogue in Ref~\cite{SICPOVM_catalogue} for $d<100$ and states in higher dimension (up to $d=\noq$), which were provided to us by Markus Grassl in a private correspondence. The construction of random POVMs is described in Appendix~ \ref{sec:numerics}.

Results of our numerical investigation are given in  Fig~\ref{fig:SIC_IC}. For every considered measurement, the success probability was obtained via direct maximization over only $\leq \nop$ random partitions of $[d^2]$.
The graph shows that with increasing dimension, $\q$ approaches $\approx 25\%$ for SIC POVMs and random POVMs, while for IC it is above $\approx 20\%$ even up to $d=\noq$.

\begin{figure}[t]
        \centering
        \includegraphics[width=0.48\textwidth]{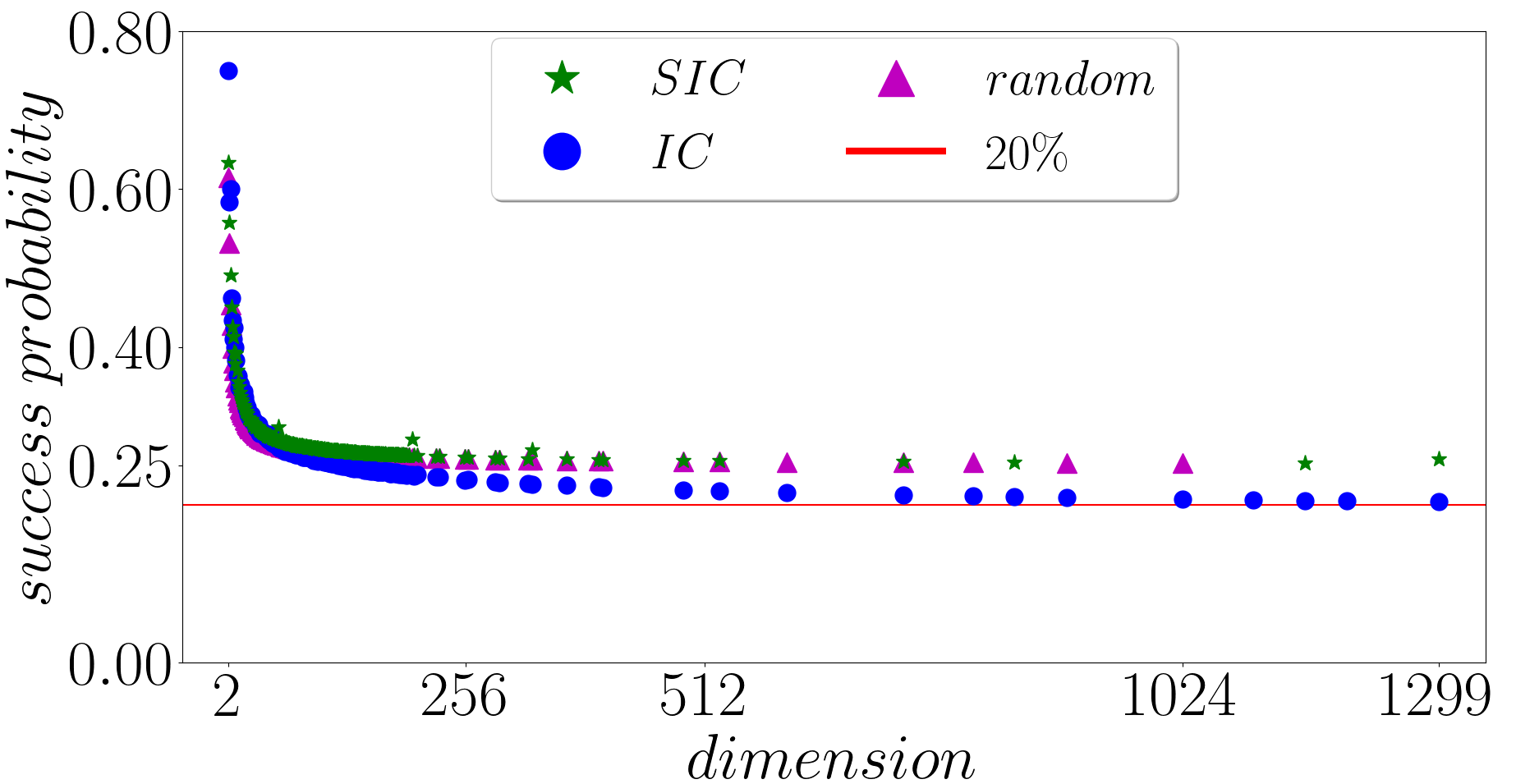}
        \caption{
        Success probability $\q$ as a function of dimension $d$ of the Hilbert space \sout{$d$} for $d^2$-outcome measurements.
     Results are shown for Weyl-Heisenberg SIC-POVMs (green stars), non-symmetric IC-POVMs (blue dots), and random POVMs (magenta triangles) for dimensions upto $\noq$. 
        For each dimension,  we plot the maximum of $\q$ (computed according to the Eq.~\eqref{eq:qsucc0}), which was obtained from random $\leq \nop$ partitions.
        For random POVMs, in each dimension, we generate \sout{from}  $10$ to $500$ random POVMs (lower number for higher dimensions) and plot the minimum $\q$ across them.
        For IC-POVMs, the measurement operators are specified by a single parameter $\alpha$ which we keep at a fixed value across all dimensions (see Appendix~\ref{sec:numerics} for details).
        }
        \label{fig:SIC_IC}
\end{figure}

\textit{Noise Analysis---}
Let us now discuss the effects of experimental imperfections  on practical implementation of our scheme for generic POVMs. The quantum circuits implementing Haar-random POVMs can be considered generic random circuits. 
The simplest noise model often adopted for such circuits (see Ref.~ \cite{Boixo2018}) is a global completely depolarizing channel described by a "visibility" parameter $\eta$. In what follows we assume that this noise is going to affect implementation of circuits used to realize a target POVM $\M$ (either via Naimark's construction or via our method). This noise acts in the following way on effects of  $n$-outcome POVM:  $M_i \rightarrow M^{\eta}_{i} \coloneqq \eta M_i +\rbracket{1-\eta}\frac{\unity}{n}$ (see Section~\ref{sec:dop} of the Appendix for details).

To quantitatively compare noisy and ideal implementation of a POVM we use Total-Variation Distance $\mathrm{d_{TV}}\rbracket{\mathbf{p}\rbracket{\M|\rho},\mathbf{p}\rbracket{\N|\rho}} \coloneqq \frac{1}{2}\sum_{i=1}^{n}|p\rbracket{i|\rho,\M}-p\rbracket{i|\rho,\N}|$ between probability distributions $\mathbf{p}\rbracket{\M|\rho}$ ($\mathbf{p}\rbracket{\N|\rho}$) obtained  when $\rho$ is measured by $\M$ ($\N$).
In particular, we will be interested in the worst-case distance, i.e., TVD \textit{maximized over quantum states} $\rho$, which can be interpreted as measure of statistical distinguishability of $\M$ and $\N$ (without using entanglement \cite{Puchala2018}). This notion of distance is used to benchmark quality of quantum measurements on near-term devices \cite{Maciejewski2020,Bravyi2020,Maciejewski2021}. 

The following result, proven in Section \ref{sec:dop} of the Appendix,  gives a \textit{lower} bound for the average worst-case distance between ideal and noisy implementation of Haar-random POVMs.

\begin{prop}\label{thm:dop_lower}
Let $\M^U$ be a Haar-random $n$-outcome rank-one POVM on $\Complex^d$ and let $\M^{U,\eta}$ be its noisy implementation with effects $\rbracket{\M^{U,\eta}}_i = \eta M^U_i + \rbracket{1-\eta}\frac{\unity}{n}$. We then have 
\begin{equation}
\label{eq:dop_lower}
\HE{ \max_{\rho} \rbracket{ \mathrm{d_{TV}}\rbracket{\pv{\M^U}{\rho},\pv{\M^{U,\eta}}{\rho}}}} \ge  \left(1-\eta\right)c_n\ ,
\end{equation}
where $c_n = \rbracket{1-\frac{1}{n}}^n\approx \frac{1}{\e}$. 
\end{prop}

To make qualitative comparison between our and standard (i.e., based on Naimark's dilation theorem) implementation of POVMs, we use noise model used in Google's recent demonstration of quantum advantage \cite{Google2019}.
Assuming that main source of errors are multiple two-qubit gates, we get that dominating term in visibility is exponentially decaying function: $\eta = \eta\rbracket{r_2,g_2} \approx \exp\rbracket{-r_2 g_2}$, where $r_2$ is two-qubit error rate and  $g_2$ is the number of two-qubit gates needed to construct a given circuit. 
Now recall that for implementation of $d^2$-outcome POVM using Naimark's dilation, one needs to implement circuits on the Hilbert space with doubled number of qubits $2N$ (we assume $d=2^N$), while our post-selection scheme requires  only  a single additional  qubit, hence the target space has only  $N+1$ qubits.
We note that for implementation of generic circuits on $2N$ qubits, the theoretical lower bound \cite{Shende2004minimal} for needed number of CNOT gates is $g^{\mathrm{Naimark}}_2=\Theta\rbracket{4^{2N}}=\Theta\rbracket{16^N}$, while our scheme gives the scaling $g^{\mathrm{post}}_2=\Theta\rbracket{4^{N}}$.

Finally, combining the above considerations with Proposition~\ref{thm:dop_lower}, we get expected worst-case distance between ideal and noisy Naimark implementation of generic $d^2$-outcome measurement is \textit{lower bounded} by $\approx \rbracket{1- \exp\rbracket{-\Theta\rbracket{16^N}}}e^{-1}$, which corresponds to $\eta^{\text{Naimark}}= \exp\rbracket{-\Theta\rbracket{16^N}}$.
We compare this to the quality of probability distribution $\mathbf{p}^{\mathrm{noise}}_{\mathrm{post}}(\M|\rho)$ generated by the noisy version of our simulation scheme  which is based on implementation of projective measurements on $N+1$ (not $2N$) qubits and hence incurring noise with  $\eta^{\text{post}}\approx \exp\rbracket{-\Theta\rbracket{4^N}}$.
In Appendix \ref{sec:dop} we show that postselection step in our scheme does not significantly affect the quality of produced samples by proving that for typical Haar random $\M^U$
\begin{equation}
    \mathrm{d_{TV}}\rbracket{\pv{\M^U}{\rho},\mathbf{p}^{\mathrm{noise}}_{\mathrm{post}}(\M^U|\rho)}\leq C (1-\exp\rbracket{-\Theta\rbracket{4^N}}),
\end{equation} where $C$ is an absolute constant. Therefore, for generic measurements, implementation via our scheme will be affected by much lower noise than in the case of Naimark's. 
We expect that similar behaviour (i.e., amount of noise in our scheme compared to Naimark's dilation) should be exhibited also for more realistic noise models -- the high reduction of the dimension of the Hilbert space is, reasonably, expected to highly reduce the noise.

\paragraph{Discussion and open problems---}

Aside from their practical relevance, our results shred light onto the question whether POVMs are more powerful (in quantum information tasks requiring sampling) than projective measurements.
Indeed, since typical POVMs in $\Complex^d$ can be implemented using $d$-outcome measurements, it suggests (and if our conjecture is true, then it implies) that, if there exists a gap in the relative usefulness (quantified for example via robustness), then it is between projective measurements and $d$-outcome POVMs. 
Moreover, the surprisingly high value of $\q^{(d)}$ will likely have potential applications to nonlocality. 
First, it  significantly limits (due to inequality \eqref{eq:rela_q_R_V})  the amount of local depolarizing noise that can be tolerated in schemes for generation secure quantum randomness using extremal $d^2$-outcome measurements \cite{Acin2016,Augusiak2020}. We also anticipate that our results can be used to construct new local models for entangled quantum states that undergo general POVM measurement (by using techniques similar to those of \cite{Oszmaniec17,Hirsch2017betterlocalhidden}).

We conclude with giving directions for future research. 
First,  naturally, is to verify whether our conjecture is true. The difficulty in proving it comes from the combinatorial nature of the optimization problem in Eq. \eqref{eq:qsucc0} - it is difficult to analytically find the optimal partition of $[n]$ that maximizes $\q$ for a target POVM $\M$. Effects of Haar random POVMs \emph{have similar properties} - in particular, they have (on average) equal operator norms - this symmetry allowed us to study them analytically. However, general POVMs can be highly unbalanced (in the sense of having effects whose operator norms can vary significantly) and suitable strategies need to be devised  to tackle such situations. Second, it is desirable to devise an algorithmic  method which, when given the circuit description of some POVM, returns the circuits needed to implement it with postselection.
Another direction is to identify and quantify the real-time implementation costs of randomisation and post-processing, and how these cost considerations can be
taken into account for suitable modifications of the scheme. Finally, it would be interesting to see if the success probability is connected to other properties of POVMs -- for instance, their entanglement cost \cite{Josza2003}.

\emph{Data availability}
The data obtained in numerical simulations is available from authors upon request.

\emph{Code availability}
The code used to obtain numerical simulations is available from authors upon request.

\emph{Acknowledgements} 
We are sincerely grateful to Markus Grassl for fruitful discussions and for sharing with us the numerical form of fiducial kets of SIC POVMs for high dimensions.
We thank Zbigniew Pucha{\l}a for the discussions at the initial stage of this project and Micha\l\ Horodecki for suggesting potential application of our scheme in PBT. The authors acknowledge the financial support by  TEAM-NET project co-financed by EU within the Smart Growth Operational Programme (contract no.  POIR.04.04.00-00-17C1/18-00). A portion of this work was done while TS was in Fudan university, and TS acknowledges support from  the National Natural Science Foundation of China (Grant No.~11875110) and  Shanghai Municipal Science and Technology Major Project (Grant No.~2019SHZDZX01).

\emph{Author Contributions}
TS had a leading role in proving Theorem 2, Proposition 1 and many auxiliary technical results. 
FBM carried out numerical simulations and proved results concerning noise robustness of POVM implementation methods. 
MO contributed with the main idea of the project, proved Theorem 1 and supervised the other parts project.
All authors equally contributed to writing the manuscript equally.

\emph{Competing interests}
The Authors declare no Competing Financial or Non-Financial Interests.

\bibliographystyle{apsrev4-2}
\bibliography{refs} 
\clearpage

\clearpage

\appendix
\onecolumngrid

\setcounter{secnumdepth}{3}
\setcounter{equation}{0}
\setcounter{figure}{0}
\setcounter{table}{0}
\makeatletter\renewcommand{\theequation}{S\arabic{equation}}
\renewcommand{\thefigure}{S\arabic{figure}} 
\renewcommand{\bibnumfmt}[1]{[S#1]}
\renewcommand{\citenumfont}[1]{#1}

\begin{center}
    \Large{\textbf{Appendix}}
\end{center}

\noindent 
\renewcommand{\thesection}{\Alph{section}}
\numberwithin{equation}{section}
 We  collect  here  technical  results  that  are  used  in  the  main  part  of  the  paper, as well as more detailed descriptions of some of the presented concepts.
In Section~\ref{sec:resources}, we discuss a relation between success probability of our implementation scheme, and a resource-theoretic quantities -- visibility and robustness of POVMs.
In Section~\ref{sec:prelim} we explain concentration of measure for general random variables on probability spaces, especially for the special cases of the unitary group $\Unitary(\bo)$ and the $(\bo-1)$-complex sphere.
The contents of this section should be treated as preliminaries for further sections.
The proofs technical version of Theorem~\ref{thm:conc_q_lower_upper} are provided in Sections~\ref{sec:conc_q_lower}.
In Section~\ref{sec:dop} we describe in more detail the effects that completely depolarizing noise has on the implementation of quantum measurements.
Finally, in Section~\ref{sec:numerics} we provide details of numerical simulations presented in the main text.

For the benefit of the reader, in table below we explain the notation used in the Appendix 

\begin{table}[!h]
\begin{centering}
\label{tab:supp}
\begin{tabular}{|c|c|}
\hline
\textbf{Symbol} & \textbf{Explanation}\tabularnewline
\hline
$d$ & dimension of principal system. \tabularnewline

$\bo$ & Number of outcomes of a target POVM. \tabularnewline

$\so$ & Number of outcomes of POVMs which we simulate target POVM with. \tabularnewline

$\Complex^d$ & Hilbert space of our principal system of study. \tabularnewline

$\Complex^{\bo}$ & Hilbert space of extended system. \tabularnewline



$\rho$, $\sigma$, etc. & General mixed states on quantum system. \tabularnewline


$\M$, $\N$, etc & Quantum measurements on our principal system. 
\tabularnewline

$\q $ & Success probability of simulating a measurement using method specified in Theorem \ref{thm:scheme} of the main text. \tabularnewline 

$R^{(\so)}(\M)$ & Robustness of a POVM $\M$ with respect to $\Sm$.  \tabularnewline

$[\bo]$, $[d]$, etc. & The set $\left\{1,2,\cdots, \bo \right\}$. Similarly for $[d]$ \tabularnewline 

$X,\ Y$, etc. &  subsets of $[\bo]$. \tabularnewline

$\Unitary(\bo)$ & Unitary group of $\bo \times \bo$ complex matrices. \tabularnewline

$U$, $W$, etc. & An $\bo \times \bo$ unitary matrix \tabularnewline

$\Unitaryx$, $\Wx$, etc. & A truncation of $\bo \times \bo$ unitary matrix $U$, occuring at the intersection of rows in $[d]$ and columns in $X$. \tabularnewline 






$\underset{U\sim \mu_\bo}{\pr} \left( \mathcal{E} \right)$ & Probability of some event $\mathcal{E}$. \tabularnewline

$\HE{f }= \intH \ f(U)$ & Integral (expectation value) of function $f$ on unitary group $\Unitary(\bo)$ with respect to the Haar measure. \tabularnewline

$|| \ \ket{\psi} \ || $ & Vector norm of state vector $\ket{\psi}$ \tabularnewline 

$\norm{A}$ & Operator norm of a linear operator $A$. \tabularnewline

$\HS{A}$ & Norm induced by Hilbert-Schmidt inner product on linear operators. \tabularnewline

$\tv (\p{p}, \p{q})$ & Total variational distance between probabilities $\p{p}$ and $\p{q}$. \tabularnewline

\hline
\end{tabular}
\par\end{centering}
\caption{Notation used in the Appendix\sout{supplementary material}}
\end{table}

\section{Relation between $q^{(\so)}$,  and critical visibility $t^{(k)}$, and robustness $R^{(k)}$, }\label{sec:resources}
\label{app:rela_q_R_V}
 Let $\Sm \subset \P(d,n)$ denote all $\so$-outcome simulable POVMs and let $\Phi_t$ denote the depolarising channel  $\Phi_{t}(X):=t X + (1-t)\frac{\tr X}{d} \unity$. 
 Its action naturally extends to POVMs, via action on individual effects: $\Phi_{t}(M_i):=t M_i + (1-t)\frac{\tr M_i}{d} \unity$. For any $\bo$-outcome POVM $\M$ the visibility with respect to $\Sm$ 
\begin{equation}
 \label{eq:visibilityM}
 t^{(\so)}\left(\M\right) \coloneqq \; \max \; \left\{ \; t \; | \; \Phi_t\left( \M \right) \in \Sm \right\}.
\end{equation}
The robustness $R^{\so}(\M)$ of measurement $\M$ with respect to $\Sm$ is defined via
\begin{equation}
 \label{eq:Rb}
 R^{(\so)}\left(\N\right) \coloneqq \; \min \; \left\{ \; s \; | \; \exists \; \mathbf{K} \; \mathrm{s.t.} \; \frac{\N + s \mathbf{K}}{1+s} \in \Sm \right\}.
\end{equation}
Let $q^{(\so)}(\M)$ be the largest success probability with which $\M$ can be simulated via $\so$-outcome POVMs. It follows that  
\begin{equation}
    \L = (q^{(\so)}(\M) M_1, q^{(\so)}(\M) M_2,\ldots, q^{(\so)}(\M) M_\bo, (1-q^{(\so)}(\M)) \unity)
\end{equation}
can be simulated via $\so$-outcome POVMs. It follows that a POVM $\Phi_{q^{(\so)}(\M)}(\M) \in\Sm$. The inequalities claimed in the main text (cf. Eq.\eqref{eq:rela_q_R_V})
\begin{equation}
    q^{(\so)}(\M) \leq t^{(\so)}(\M)\ ,\ R^{(\so)}(\M) \leq \frac{1}{q^{(\so)}(\M)} -1 \ .
\end{equation}
follow directly from definitions of $t^{(\so)}(\M)$ and $R^{(\so)}(\M)$ respectively.

\section{Preliminaries}
\label{sec:prelim}

In this Part we provide some basic theoretical background that will be used in Lemmas  \ref{lem:lip_normUj}, \ref{lem:lip_Pz}, \ref{lem:UjavgbigO}, \ref{lem:Nr}, \ref{lem:Dr}, and Theorems \ref{thm:conc_q_lower_technical} and \ref{thm:conc_q_upper_technical}. In Subsection \ref{subsec:conc_measure}, we introduce the notion of concentration of measure, which will be used extensively for proving the aforementioned lemmas and theorems. Related concepts like Lipshitz constants of functions and log-Sobolev inequalities and log-Sobolev constants are also explained alongside. The metric spaces which we use in this work are the unitary group  $\Unitary(\bo)$ (with metric induced by Hilbert-Schmidt inner product), and the $(\bo-1)$-complex sphere $\Sc{\bo-1}$, with the metric it inherits from $\Complex^{\bo}$. The Haar-measure on $\Unitary(\bo)$ and the uniform measure on $\Sc{\bo-1}$ will be introduced in subsections \ref{subsec:Haar_measure} and \ref{subsec:uniform_measure_Sc} respectively, and the corresponding log-Sobolev constants also mentioned.  

\subsection{Concentration of Measure: Lipshitz constants and log-Sobolev inequalities}
\label{subsec:conc_measure}
We start by recalling notions of Lipshitz constants and log-Sobolev inequalities.
Let $(\mathcal{X},d)$ be a metric space, and let $f:\mathcal{X} \rightarrow \R$ be a real function on $\mathcal{X}$. We say that $f$ is $L$-Lipshitz on $\mathcal{X}$ with respect to the metric $d$, if $f$ satisfies the following condition.
\begin{equation}
\label{eq:Lip_def}
|f(x) - f(y)| \le L \ d(x,y), \; \mathrm{for} \; \mathrm{all} \; x,y \in \mathcal{X}.
\end{equation} \noindent\noindent
Now let $\mu$ be a probability measure on $(\mathcal{X},d)$, and let function $f$ be such that the length of the gradient of $f$ can be defined at any point $x$ in $\mathcal{X}$, namely
\begin{equation}
 \label{eq:gradf}
 \left| \nabla f\right|(x) :=\mathrm{lim} \; \mathrm{sup}_{y \rightarrow x} \; \dfrac{|f(x) - f(y)|}{d(x,y)}.
\end{equation} \noindent \noindent
Then for any such function, the following \textit{concentration inequalities} hold
\begin{equation}
    \label{eq:log_Sob1}
  \int \; d\mu(x) \; \mathrm{exp}\left( \lambda \left( f(x) - \int d \mu(x) f(x) \right) \right)   \le \mathrm{exp} \left( \dfrac{C L^2 \lambda^2}{2}\right), \; \mathrm{for} \; \mathrm{all} \; \lambda \in \R, 
\end{equation} \noindent
\begin{equation}
    \label{eq:log_Sob2}
    \underset{x\sim \mu}{\pr} \left( \, f(x) \; \ge \; \int d\mu(x) f(x) + t \; \right) \le \mathrm{exp} \left( - \dfrac{t^2}{2 C L^2}  \right), \; \mathrm{for} \; t \ge 0,
\end{equation} \noindent \noindent
where $C$ is called the log-Sobolev constant of $\mu$ with respect to the metric $d$ of $\mathcal{X}$. 
We note that the inequality \eqref{eq:log_Sob2} can be derived from \eqref{eq:log_Sob1} (see Theorem 5.39, in \cite{Aubrun2017}).
We refer the reader to \cite{Aubrun2017} for more details on log-Sobolev inequalities. 

\subsection{Haar-measure on \texorpdfstring{$\Unitary(\bo)$}{}}
\label{subsec:Haar_measure}
The group of $\bo \times \bo$ unitary matrices $\Unitary(\bo)$ is endowed with the well known probability measure known as the Haar-measure. 
It follows that for any integrable function $f$ on $U(n)$, its expectation value with respect to the Haar measure is invariant under the following operations
\begin{equation}
 \label{eq:Haarmeasure}
 \intH  f (U)
 = \intH  f (W U)
 = \intH  f (U W)
 = \intH f \left( U^{-1} \right), 
\end{equation} \noindent \noindent
where $W$ is an arbitrary fixed unitary in $U(n)$.
$\Unitary(\bo)$ inherits a metric from the Hilbert-Schmidt inner product on the space of $\bo \times \bo$ complex matrices.
The distance between two unitaries $U,W$ with respect to the Hilbert-Schmidt metric is 
\begin{equation}
 \label{eq:HS_metric}
 \HS{U-W} = \sqrt{ \mathrm{tr} \left( (U-W)^\dag (U-W) \right) }.
\end{equation} \noindent \noindent
The follwing Theorem then gives the log-Sobolev constant for the Haar measure with respect to the Hilbert-Schmidt metric (table 5.4 in \cite{Aubrun2017}).

\begin{thm}\label{thm:log_Sob_const_Haar}\cite{Aubrun2017}
 The log-Sobolev constant for the Haar measure on the unitary group $\Unitary(\bo)$ with the Hilbert-Schmidt metric is $\dfrac{6}{\bo}$.
 \end{thm}
 
 \subsection{Uniform measure on \texorpdfstring{$\Sc{\bo-1}$}{}}
 \label{subsec:uniform_measure_Sc}
The complex $(\bo-1)$-sphere $\Sc{\bo-1}$ is defined as
\begin{equation}
 \label{eq:defSc}
 \Sc{\bo-1} = \left\{ \ket{x} \in \Complex^{\bo} \; | \; \bk{x}{x} =1 \right\}.
\end{equation} 
For any $\bo \times \bo$ unitary $U$, the unitary action $\ket{x} \rightarrow U \ket{x}$ is norm-preserving. Thus,  the Haar-measure of $U(n)$ endows a rotationally invariant probability measure on $\Sc{n-1}$ in the following way: fix some arbitrary $\ket{x}$ in $\Sc{\bo-1}$, then for Haar-random $U$, $\ket{z} = U \ket{x}$ is a random variable in $\Sc{\bo-1}$, endowed with a probability measure called the uniform probability measure on $\Sc{\bo-1}$. In particular, one can choose $\ket{x}$ to be a standard basis vector $\ket{e_i}$, which tells us that when $U$ is Haar-random, then it's columns are distributed with the uniform measure on $\Sc{\bo-1}$. The uniform probability measure on $\Sc{\bo-1}$ has a log-Sobolev constant with respect to the usual norm-induced metric on $\Sc{\bo-1}$ (see table 5.4 in \cite {Aubrun2017}; note that $\Sc{\bo-1} \simeq \Sr{2\bo-1}$, which is the $(2\bo-1)$-sphere in $\R^{2\bo}$).
 
 \begin{thm}
 \label{thm:log_Sob_const_Sc}
 \cite{Aubrun2017} The log-Sobolev constant for the uniform measure on the complex $(\bo-1)$-sphere, $\Sc{\bo-1}$ is $\dfrac{1}{2\bo-1}$.
\end{thm}


Let $\left\{ \ \ket{e_j} \ \right\}_{j=1}^{\bo}$, denote the standard basis for $\Complex^{\bo}$. Each vector $\ket{\psi}$ in $\Sc{\bo-1}$ can be mapped to an $n$-probability vector as follows:
\begin{equation}
 \label{eq:x_p}
\ket{\psi} \, \rightarrow \, \mathrm{\bf{p}}, \; \mathrm{where} \; p_i = |\bk{e_i}{\psi}|^2.
\end{equation}
Imposing the uniform measure on $\Sc{\bo-1}$, converts $p_i$ into a random variable on interval $[0,1]$. Denote $p_i$ by $x$, the probability density of this random variable is given by \cite{Zyczkowski2000}
\begin{equation}
 \label{eq:px}
 p(x) = (n-1) (1-x)^{n-2}, \; \mathrm{where} \; 0 \le x \le 1.
\end{equation} \noindent \noindent
It is easy to see that the expectation value of $x$ is $\frac{1}{n}$. Also, the distribution of $x$ is given by $
 \pr \left( x \ge y \right) \; = \; \left( 1 - y \right)^{n-1}$ and it follows that
\begin{equation}
 \label{eq:mua_approx}
 \pr \left( x \ge y \right) \; \le \; \mathrm{exp} \left( - (n-1) y \right).
\end{equation} \noindent \noindent

\subsection{Haar-random POVMs}
\label{subsec:Haar_random_POVM}

In this subsection we recall the construction of rank-one Haar random POVMs. An $\bo$-outcome, rank-one POVM $\M^U$ on $\Complex^{d}$ can be constructed from Haar-random unitary $U\in\Unitary(\bo)$ using the following steps
\begin{enumerate}
 \item[1.] Extend the principal system $\Complex^{d}$ to a larger system $\Complex^{\bo}$ using an ancillary system, which is prepared in a fixed state $\ket{0}$.
 \item[2.] Rotate the composite system by the unitary $U$ in $\Unitary(\bo)$. 
 \item[3.] Measure the composite system in a computational basis $\cbracket{\ket{e_i}}_{i=1}^{n}$. 
 \end{enumerate}
Let us denote by $\PPP^{U}$ a rank-1 $\bo$-outcome projective measurement on the composite system, whose effects are given by
\begin{equation}
 \label{eq:PU_effects}
 P^U_i  \, =  \, U^\dag \ \kb{e_i}{e_i} \ U, \; \mathrm{for} \; i \in [\bo].
\end{equation} \noindent
Now if the ancillary system is prepared in state $\ketbra{0}{0}$, then performing the above measurement on the composite system, implements on original system $C^{d}$ a rank-1 $\bo$-outcome measurement $\M^{U}$ with effects given by $M^{U}_i = \tr_B \rbracket{\unity \otimes \kb{0}{0} \ U^{\dag} P_i U }$.
Importantly, the matrix elements of $M_i$ can be related to the matrix elements of $U$ via
\begin{equation}
 \label{eq:effects_U_matrix_elements}
 \left( \ M_j \right)_{il} \; = \; U_{ji}^* \ U_{jl}.
\end{equation} \noindent

Finally, when $U$ is distributed according to the Haar measure on $\Unitary(\bo)$, then a POVM $\M^U$ also becomes a random variable. This is called a Haar-random POVM.

\section{Proof of Theorem \ref{thm:conc_q_lower_upper} }
\label{sec:conc_q_lower}

In this section we prove the Theorem~\ref{thm:conc_q_lower_upper} concerning bounds on the success probability of implementation of Haar-random POVMs with postselection. The first three subsections contain auxiliary lemmas needed in the proof of the main result which we provide in Section \ref{app:technicalBounds}. 
From now on, unless stated otherwise, we denote by $X$ a subset of $\sbracket{\bo}$ such that $|X|=\so$, by $U$ a $\bo\times\bo$ unitary matrix, and by $U_{X}$ a truncation of unitary $U$, occurring at the intersection between rows in $[d]$ and columns in $X$.
Furthermore, $\cbracket{\ket{e_i}}_{i}^{\bo}$ is a standard orthonormal basis in $\Complex^{\bo}$ and by $P = \sum_{i=1}^d \kb{e_i}{e_i}$ we denote a projector onto the space of its first $d$ components.

\subsection{Lipshitz constants for functions used in proof of Theorem \ref{thm:conc_q_lower_upper}}
\label{subsec:lower}
We first bound Lipshitz constants for some functions which will be used latter.
\begin{lem}
 \label{lem:lip_normUj}
The function $U \rightarrow \norm{\Unitaryx}$ is $1$-Lipshitz on $\Unitary(\bo)$ with respect to the Hilbert-Schmidt metric.
 \end{lem}
\begin{proof}
Let $U,W$ be two $\bo \times \bo$ unitaries, such that $U \neq W$.
Then
\begin{equation}
    \label{eq:norm_Lip}
    \dfrac{\left| \;\norm{\Unitaryx} - \norm{\Wx} \; \right|}{\HS{U-W}} \le \dfrac{\norm{\Unitaryx - \Wx}}{\HS{U-W}} \le \dfrac{\HS{\Unitaryx - \Wx}}{\HS{U-W}} \le 1.
\end{equation} \noindent \noindent
\end{proof}
\begin{lem}
 \label{lem:lip_Pz}
For any $\ket{z}$ in $\Sc{n-1}$, the function $\ket{z} \rightarrow || P \ket{z} ||$ is $1$-Lipshitz. 
\end{lem}
\begin{proof}
 Let $\ket{z_1}, \ket{z_2} \in \Sc{\bo-1}$, such that $\ket{z_1} \neq \ket{z_2}$. Then
\begin{equation}
 \label{eq:PZ_lip}
 \dfrac{| \; \np{z_1} - \np{z_2} \; |}{|| \; \ket{z_1} - \ket{z_2} \; ||} \le \dfrac{|| P \left( \ket{z_1} - \ket{z_2} \right) \; ||}{|| \; \ket{z_1} - \ket{z_2} \; ||} \le  \norm{P} = 1.
\end{equation} \noindent \noindent
\end{proof}

\noindent

\subsection{Upper bound to the Haar-averaged norm of truncations of unitary matrices}
\label{subsec:UjavgbigO}
The following auxiliary  results allow us to upper bound expected value of the operator norm of truncations of Haar random unitaries.

\begin{lem}
\label{lem:avgUj_discrete}
Let $\Sx \subset \Sc{\bo-1}$ be defined as 
\begin{equation}
 \label{eq:S_j}
\Sx = \left\{ \ket{a} \in \Sc{\bo-1} \; | \; \bk{e_i}{a} = 0, \; \forall \; i \notin X \right\},
 \end{equation} \noindent \noindent
Let $\Ej$ be an $\epsilon$-net for $\Sx$.
Then $ \bra{e_i} P \ket{x}=0$ for $i\ge d+1$ for all $\ket{x} \in \Complex^\bo$, and we have
\begin{equation}
\label{eq:mean_of_Uj_op_norm_vs_mean_of_ReyUx}
 \norm{\Unitaryx}  \le \dfrac{1}{1 -  \epsilon} \left( \underset{ \sk{x} \in \Ej}{\max} \; \nup{x} \right).
\end{equation} \noindent \noindent
\end{lem}

\begin{proof}
From the singular value decomposition of $\Unitaryx$, we get that 
\begin{equation}
 \label{eq:op_norm_Uj_using_SVD2}
\norm{\Unitaryx} = \underset{\sk{a} \in \Sx}{\max} \, \nup{a} =  \nup{\tilde{a}}, 
\end{equation} \noindent \noindent
where $\ket{\tilde{a}} \in \Sx$ is the (or is a) vector at which the maximization in equation \eqref{eq:op_norm_Uj_using_SVD2} is attained.

Now to discretize the optimization in equation \eqref{eq:op_norm_Uj_using_SVD2}, we optimize over $\Ej$ instead, and we note that then there exists $\ket{\tilde{x}} \in \Ej$ such that $|| \; \ket{\tilde{x}} - \ket{\tilde{a}} \; ||$ $\le \epsilon$. Hence we get that $\nup{\tilde{a}} \le \nup{\tilde{x}} + \epsilon \  \norm{\Unitaryx} $, which gives us 
\begin{align}
 \label{eq:using_epsilon_nets}
    \norm{\Unitaryx} \le \dfrac{1}{1 - \epsilon} \nup{\tilde{x}} \le \, \dfrac{1}{1-\epsilon} \ \left( \underset{\sk{x} \in \Ej}{\max} \; \nup{x} \right),
\end{align}
 for all $0 < \epsilon < 1$. 
\end{proof}

\begin{lem}
\label{lem:UjavgbigO}
We have the following upper bound for expected value of the norm of truncation of the unitary matrix
\begin{equation}
 \label{eq:avgUibigO}
 \HE{ \, \norm{\Unitaryx} \, } \; \le c \ \left(1+\sqrt{\frac{2\m}{d}}\right) \; \sqrt{\dfrac{d}{\bo}}, \; \mathrm{where} \; c \approx  1.92.
\end{equation} \noindent \noindent
Additionally, when $\m=d-1$, 
\begin{equation}
 \label{eq:avgUibigO_mequalsd}
 \HE{ \, \norm{\Unitaryx} \, } \; \le c \ \sqrt{\dfrac{d-1}{\bo}}, \; \mathrm{where} \; c \approx  3.86\ . 
\end{equation} \noindent \noindent
\end{lem}

\begin{rem}
 \label{rem1}
The proof of Lemma \ref{lem:UjavgbigO} is inspired by the proof of equation (18) and  Theorem 7 in \cite{RRZK16} (please see Section 2 of the appendix in \cite{RRZK16}). 
In Remark \ref{rem2} below we briefly explain the differences between the proof presented here and the proof in \cite{RRZK16}.
\end{rem}

\begin{proof}
Let $\ket{z} \in \Sc{\bo-1}$, and define the function $\ket{z} \, \rightarrow \, \np{z}$. 
This function is $1$-Lipshitz on $\Sc{n-1}$ (Lemma \ref{lem:lip_Pz}). 
Define $\Sx$ as in equation \eqref{eq:S_j}. 
Now fix some $\ket{x} \in \Sx$. Let $U \in \Unitary(\bo)$ be Haar-random, and let $\ket{z} = U \ket{x}$.
Then $\ket{z}$ is uniformly distributed on $ S^{\bo-1}_{\Complex}$ (see Subsection \ref{subsec:uniform_measure_Sc}
).
Thus the function $\ket{z} \rightarrow \np{z}$ satisfies the following log-Sobolev inequality with a constant $C=\frac{1}{2\bo-1}$, with respect to the uniform measure on $\Sc{\bo-1}$ (
see Subsection \ref{subsec:uniform_measure_Sc}
)
\begin{align}
 \label{eq:log-Sobolev3_1}
 \intH \mathrm{exp}\left(  \ \lambda \left( \ \nup{x}  - A \right)  \     \right)  \le \exp\left( \frac{\lambda^2}{2 (2\bo-1)}\right), \; \forall \; \lambda \in \R.
\end{align} \noindent
where $A \coloneqq \HE{ \ \nup{x} \ }$. Since $\exp\left(- \lambda \ A\right)$ is independent of the integrating variable, we get
\begin{align}
 \label{eq:log-Sobolev3_2}
 \intH \mathrm{exp}\left( \ \lambda \  \nup{x}   \  \right)  \le \mathrm{exp}\left( \frac{\lambda^2}{2 (2\bo-1) } + \lambda A\right).
\end{align}
First we prove that $A \le \sqrt{\frac{d}{\bo}}$.
Using the well-known result $\HE{|U_{ij}|^2}=\frac{1}{n}$, one obtains
\begin{equation}
 \label{eq:Haar_avg_norm_sq_PUketx}
 \HE{\norm{PU \ket{x}}^2} = \sum_{i=1}^n \HE{|U_{ij}|^2} = \dfrac{d}{n},
\end{equation} \noindent \noindent
where we chose $\ket{x} = \ket{e_j}$ for some $j \in X$. Now note that $\HE{\norm{PU \ket{x}}^2} \ge \HE{ \  \norm{PU \ket{x}} \ }^2$.
Hence we get 
\begin{align}
 \label{eq:log-Sobolev3}
 \intH \mathrm{exp}\left( \ \lambda \  \nup{x}   \   \right)  \le \mathrm{exp}\left( \frac{\lambda^2}{2 (2\bo-1)}  +  \lambda \sqrt{\frac{d}{\bo}}\right).
\end{align}
Now let $\Ej$ be an $\epsilon$-net for $\Sx$.
Then we sum the inequality \eqref{eq:log-Sobolev3} over all $\ket{x} \in \Ej$, and we get
\begin{equation}
 \label{eq:log-Sobolev5}
 \sum_{\sk{x} \in \Ej} \intH \; \mathrm{exp} \left(  \ \lambda \  \nup{x}   \  \right)  \;  \le |\Ej| \;  \mathrm{exp}\left( \frac{\lambda^2 }{2 (2\bo-1)}  +  \lambda \sqrt{\frac{d}{\bo}} \right).
\end{equation} \noindent 
For each $U \in \Unitary(\bo)$ there is some $\ket{x_{\scriptscriptstyle{U}}} \in \Ej$, such that 
\begin{equation}
\label{eq:maxReyiUxj}
 \nup{x_{\scriptscriptstyle{U}}} = \underset{  \sk{x} \in \Ej }{\max} \; \nup{x}.
\end{equation} 
It is not difficult to see that $U \rightarrow \nup{x_{\scriptscriptstyle{U}}}$ is a continuous function, which implies that $\exp \left(\lambda \ \nup{x_{\scriptscriptstyle{U}}} \ \right)$ is integrable on $\Unitary(\bo)$. Thus we get
\begin{align}
\label{eq:mean_of_complicated_function}
\intH \mathrm{exp} \left( \ \lambda \  \nup{x_{\scriptscriptstyle{U}}}   \ \right) 
\le  &  \; \sum_{\sk{x} \in \Ej  }      \HE{ \ \mathrm{exp} \left( \ \lambda \  \nup{x}  \ \right) \ } \notag \\
\le  & \; |\Ej| \; \mathrm{exp}\left( \frac{\lambda^2  }{2\left(2\bo-1 \right)} +  \lambda \sqrt{\frac{d}{\bo}} \right).
 \end{align}
Since the exponential function is convex, Jensen's inequality can be applied in equation \eqref{eq:mean_of_complicated_function}, which gives
\begin{equation}
 \label{eq:exp_of_mean_1}
 \mathrm{exp} \; \left( \lambda \intH \;  \nup{x_{\scriptscriptstyle{U}}} \right) \; \le \; |\Ej| \; \mathrm{exp}\left( \frac{\lambda^2 }{2\left(2\bo-1 \right)} +  \lambda \sqrt{\frac{d}{\bo}} \right) \ .
 \end{equation} \noindent \noindent
 Now taking the (natural) logarithm (and assuming that $\lambda > 0 $) we get
 \begin{equation}
  \label{eq:exp_of_mean_2} 
  \intH \;  \nup{x_{\scriptscriptstyle{U}}}   \le \frac{1}{\lambda } \; \left( \log|\Ej| + \frac{\lambda^2 }{2\left(2\bo-1 \right)} +  \lambda \sqrt{\frac{d}{\bo}}  \right).
\end{equation} \noindent \noindent
Since the inequality \eqref{eq:exp_of_mean_2} is valid for all $\lambda > 0$, we directly minimize the RHS over $\lambda$, and we get
\begin{align}
 \label{eq:Uj_norm_mean1}
 \HE{ \ \underset{\sk{x} \in \Ej}{\max} \, \nup{x} \ \  } \le \sqrt{\dfrac{2 \ \log |\Ej| \ }{2\bo-1}} +  \sqrt{\frac{d}{\bo}},
\end{align}
which is obtained at the value $\lambda = \sqrt{2 (2n-1) \ \log|\Ej|}$. Note that we have used equation \eqref{eq:maxReyiUxj} in the LHS of equation \eqref{eq:Uj_norm_mean1}. 

There's a well-known theorem (see, e.g., \cite{Szarek1998, Aubrun2017})
that an $\epsilon$-net for $\Sx$ has at most $\left( 1 + 2/\epsilon\right)^{2\m}$ points. 
This gives us an upper bound for $| \Ej|$, which inserted into RHS of \eqref{eq:Uj_norm_mean1} gives
\begin{align}
 \label{eq:Uj_norm_mean2}
 \HE{ \underset{\sk{x} \in \Ej}{\max} \; \nup{x} \;  } \le \sqrt{\dfrac{4 \ \m \ \log \ (1+2/\epsilon)}{2\bo-1}} +   \sqrt{\frac{d}{\bo}} \ .
\end{align}
Finally, using the inequality \eqref{eq:mean_of_Uj_op_norm_vs_mean_of_ReyUx} with the inequality \eqref{eq:Uj_norm_mean2} we get 
\begin{equation}
 \label{eq:mean_of_Uj_norm1_1}
 \norm{\Unitaryx}  \le \frac{1}{1 - \epsilon} \left( \sqrt{ \dfrac{4 \ \m \  \log \ (1 + 2/\epsilon) \ }{2\bo-1} } + \sqrt{\dfrac{d}{\bo}}\right), \; \mathrm{for} \; \epsilon \in (0,1) ,
\end{equation} \noindent
which is valid for any $\epsilon \in \rbracket{0,1}$.
Now recall that in our scheme we are interested in case when $\m \le d-1$, which allows to rewrite the above inequality as
\begin{equation}
 \label{eq:mean_of_Uj_norm1}
 \norm{\Unitaryx} \,  \le \, \frac{\sqrt{ \ \log \ (1 + 2/\epsilon) \ } \ }{1 - \epsilon} \, \left( \  \ 1 + \sqrt{  \dfrac{2 \m }{d} \ }  \ \right) \  \sqrt{ \ \dfrac{d}{\bo} \  }, \;  \mathrm{for} \; \epsilon \in (0,1),
\end{equation} \noindent \noindent
where we have used the fact that $1 < \log(1+2/\epsilon)$ 
for any $\epsilon \in (0,1)$ and we assume that $n$ is large.
With this approximation it is possible to perform minimization over $\varepsilon$, which gives us the inequality \eqref{eq:avgUibigO}.
Note that the result of minimization will generally depend on the relative values of $d$ and $\m$, and so for special case $\m=d-1$ we get inequality \eqref{eq:avgUibigO_mequalsd}.
\end{proof}

\begin{rem}
\label{rem2}
There are two differences between the proof that we gave above, and the proof for equation (18) in Theorem 7 of \cite{RRZK16}. Firstly, the goal of Lemma \ref{lem:UjavgbigO} is to find an upper bound to $\HE{ \ \norm{\Unitaryx} \ }$, while in \cite{RRZK16}, the upper bound being sought is for $\HE{ \ \underset{X,Y}{\max} \, \norm{U_{\scriptscriptstyle X,Y}} \ } $,  where $X, Y \in [\bo]$, such that $|X| = d$ and $|Y|=\m$, and $U_{\scriptscriptstyle{X,Y}}$ is the $d \times \m$ truncation of $U$ lying at the intersection between rows in $X$ and columns in $Y$. For this purpose, the optimization in \cite{RRZK16} is over an $\epsilon$-net whose cardinality is $\binom{\bo}{d}\binom{\bo}{\m} \left( 1 + \frac{2}{\epsilon} \right)^{2(d + \m)}$. The second difference is that we use the equation \eqref{eq:mean_of_Uj_op_norm_vs_mean_of_ReyUx} for the optimization, whereas in \cite{RRZK16}, they used $\norm{U_{\scriptscriptstyle{X,Y}}} \ = \ \underset{\sk{x},\sk{y}}{\max} \; \mathrm{Re} \ \bra{x} U \ket{y} $, where $\ket{x} \in E_{\scriptscriptstyle{X}}$, which is an $\epsilon$-net for $\Sx$, and $\ket{y} \in E_{\scriptscriptstyle{Y}}$, which is an epsilon-net for $S_{\scriptscriptstyle{Y}}$. Our reason for choosing equation \eqref{eq:mean_of_Uj_op_norm_vs_mean_of_ReyUx} is that it allows us to obtain a lower upper bound in inequality \eqref{eq:avgUibigO} and \eqref{eq:avgUibigO_mequalsd}. This is important because this upper bound is closely associated with the success probability, as can be seen in the proof of Theorem \ref{thm:conc_q_lower_technical}.
\end{rem}

\subsection{Lemmas needed for the proof of upper bound on $q^{(\so)}(\M^U)$}

\begin{lem} 
 \label{lem:Nr}
Let $\all{j}{U} := \sum_{i=1}^d |U_{ij}|^2$ for $j \in [\bo]$. Then we have 

\begin{equation}
 \label{eq:conc_measure_g}
 \underset{U\sim \mu_\bo}{\pr}\left( \max_{X\subset [\bo], |X|\leq \so}\sum_{j \in X} \all{j}{U} \; \ge \; \dfrac{2d\so\log \bo}{\bo} \,  \left( 1+ \epsilon \right) \, \right) \; \le \;  \; \dfrac{\e \; d}{\bo^{1+2\epsilon}}, \; \mathrm{where} \; \epsilon \in \left( \dfrac{1}{2 \log n}, \dfrac{\bo}{2 \log \bo} -1 \right).
\end{equation} \noindent
\end{lem}

\begin{proof}
 Consider the event
 \begin{equation}
  \label{eq:event1}
  \mathcal{E}:= \; \left\{ \ U \in \Unitary(\bo) \ \Bigg| \ \max_{X\subset [\bo], |X|\leq \so} \ \sum_{j \in X} \all{j}{U} \ \ge \ \dfrac{r d \so  \log \bo}{\bo} \left( 1+ \epsilon \right) \right\}, 
 \end{equation} \noindent
 where $r$ is a constant, 
 that will be determined later to get a decent concentration. The event $ \mathcal{E}$ implies that there exists some $i \in [d]$ and some $j \in [\bo]$ such that the following event is true:
 \begin{equation}
  \label{eq:eventij}
  \mathcal{E}_{ij} := \left\{ U \in \Unitary(\bo) \ \Bigg | \ |U_{ij}|^2 \ \ge \  \dfrac{ r \log \bo}{\bo} (1+\epsilon) \right\}.
 \end{equation} \noindent
 Hence we have
 \begin{equation}
       \label{eq:E_in_union_Eij}
       \mathcal{E} \, \subset \bigcup_{\substack{ i \in [d], j \in [\bo]}} \ \mathcal{E}_{ij}.
      \end{equation} \noindent
Now we note that for $ \epsilon \in \left(-1,\frac{\bo}{r \log \bo} -1 \right)$ and $y= \dfrac{ r (1+\epsilon) \log \bo }{\bo} $, from inequality \eqref{eq:mua_approx} we have
 \begin{equation}
  \label{eq:muy_1}
  \underset{U\sim \mu_\bo}{\pr} \left( \mathcal{E}_{ij} \right) \; \le \; \dfrac{\e}{\bo^{r(1+\epsilon)}},
 \end{equation} \noindent
where we used the fact that $1 < \exp \left( \frac{r(1 + \epsilon) \log \bo}{\bo} \right) < \e$. 
Using the union bound gives
\begin{equation}
 \label{eq:unionbound1}
 \underset{U\sim \mu_\bo}{\pr} \left(\; \bigcup_{\substack{i \in [d], j \in [\bo]}} \; \mathcal{E}_{ij} \, \right) \le \; \sum_{i \in [d], j \in [\bo]} \; \underset{U\sim \mu_\bo}{\pr} \left( \mathcal{E}_{ij}\right) \; \le  \; \dfrac{ \e \ d \ \bo}{\bo^{r(1+\epsilon)}}.
\end{equation} \noindent
Note that for the probabilities appearing on the RHS of the inequality \eqref{eq:unionbound1} to be meaningful, it's necessary to revise the interval for $\epsilon$ as follows.
\begin{equation}
 \label{eq:epsilon_interval}
 \dfrac{1}{r \log \bo} \ < \ \epsilon \ < \dfrac{n}{r \log n} -1,
\end{equation} \noindent
provided that $r$ is chosen so that ${\bo}^r \ge \bo \;  d$. The maximum value of $d$ in terms of $n$ is when $d = n$. Thus we choose $r=2$, which proves the lemma.
\end{proof}

\begin{lem}
\label{lem:Dr}
Let $\all{j}{U} := \sum_{i=1}^d |U_{ij}|^2$ for $j \in [\bo]$.  Then 
\begin{equation}
 \label{eq:Dr} 
 \underset{U \in \Unitary(\bo)}{\min} \; \sum_{j=1}^{\bo}\, (\all{j}{U})^2 \; = \; \frac{d^2}{\bo}.
\end{equation} \noindent
\end{lem}
\begin{proof}
Since $\sum_{j=1}^{\bo} \all{j}{U} = d$, we get that $\left( \frac{w_1}{d},\frac{w_2}{d},\cdots,\frac{w_{\bo}}{d}\right)$ (where we dispense with the superscript $U$) is an $\bo$-probability vector.
For any $n$-probability $\mathrm{\bf{p}}$, consider the function $\mathrm{\bf{p}} \rightarrow \sum_{j=1}^{\bo} p_j^2$ is Schur-convex \cite{bengtsson_zyczkowski_2006} and therefore its minimum value is
\begin{equation}
 \label{eq:sum_pjsquared}
  \underset{\mathrm{\bf{p}}}{\min} \; \sum_{j=1}^{\bo} \, p_j^2 \,= \, \frac{1}{n},
\end{equation} \noindent \noindent
where minimization goes over all $\bo$-probability vectors, and is attained at the uniform distribution, $\mathrm{\bf{p}} = (\frac{1}{n},\frac{1}{n},\cdots,\frac{1}{n})$. Finally, to prove the lemma we note that for the Fourier matrix $F$, with elements
\begin{equation}
 \label{eq:Fourier}
 F_{jl} = \frac{1}{\sqrt{\bo}}\omega^{(j-1)(l-1)} , \; \mathrm{where} \; \omega = \exp{\frac{2 \pi i}{\bo}}, \; \mathrm{and} \; j,l \in[\bo]\, 
\end{equation} \noindent
we have $\sum_{j=1}^{\bo} (\all{j}{F})^2=\frac{d^2}{n}$.
\end{proof}

\subsection{Technical version Theorem \ref{thm:conc_q_lower_upper} in the main text}
\label{app:technicalBounds}

Now we are ready to provide a technical version of the first part of the Theorem~\ref{thm:conc_q_lower_upper} from the main text. Since the methods used in the proofs of  inequalities ~\eqref{eq:conc_q_lower_simpler1} and ~\eqref{eq:conc_q_upper_simpler1} comprising Theorem ~\ref{thm:conc_q_lower_upper}  differ, we formulated two auxiliary technical theorems (Theorem \ref{thm:conc_q_lower_technical} and Theorem \ref{thm:conc_q_upper_technical} below), each covering one of the inequalities.

\begin{thm}
\label{thm:conc_q_lower_technical}[Technical formulation of inequality~\eqref{eq:conc_q_lower_simpler1} from Theorem \ref{thm:conc_q_lower_upper}]
Let $\bo \in \left\{ d,\ldots,d^2 \right\} $, $\so\leq d$. Let $\M^U$ denote a rank-one $\bo$-outcome Haar-random POVM on $\Complex^d$. Let $\q^{(\so)}(\M^U)$ denote success probability of implementing $\M^U$ via $\so$-outcome measurements as in Eq. \eqref{eq:qsucc0} for the standard partition $X_1=\left\{ 1,\ldots  \so-1 \right\},\ X_2=\left\{ \so,\so+1,\ldots, 2\so-2 \right\}$, etc. , of $[\bo]$. We then have

\begin{equation}
 \label{eq:conc_q_lower_technical}
   \underset{U\sim \mu_\bo}{\pr}\left( \ \q^{(\so)}(\M^U) \ge  c \dfrac{\gamma}{\left(1+\sqrt{\gamma}\right)^2} \left( 1 - \epsilon \right)  \right)  \ge \; 1  \ - \  \frac{\bo}{\so-1} \; \mathrm{exp} \left( -a \left( 1+\sqrt{\gamma} \right)^2 d \; \epsilon^2 \right),
\end{equation} \noindent \noindent  
where $ 0 < \epsilon < \frac{\sqrt{5}-1}{2}$, $\gamma = \frac{2 (m-1)}{d}$, $c \approx 6.79 \times 10^{-2}$ and $a \approx 0.307$. Furthermore, for special case $\so=d$, we have
\begin{equation}
 \label{eq:conc_q_lower_mequalsd}
   \underset{U\sim \mu_\bo}{\pr}\left( \ \q^{(d)}(U) \ge \  c \left( 1 - \epsilon \right) \ \right)  \ge \; 1 \  - \ \frac{\bo}{d-1} \ \mathrm{exp} \left( \ - a \ d \ \epsilon^2 \ \right), \; \mathrm{where} \;  0 < \epsilon < 1,
\end{equation} \noindent \noindent
where $c \approx 6.74 \times 10^{-2}$ and $a \approx 1.79$.
\end{thm}
\begin{rem}
 \label{rem3}
 One can directly obtain an upper bound for the $\so=d$ case, by evaluating the RHS of inequality \eqref{eq:conc_q_lower_technical} for $\so=d1$. But in that case the success probability one gets is $4.65\%$, which is lower than the success probability in inequality \eqref{eq:conc_q_lower_mequalsd} (which is ($6.74\%$). Thus, a separate derivation for \eqref{eq:conc_q_lower_mequalsd} is warranted.
\end{rem}

\begin{proof}
Let $U_j$ be a truncation of $U$, occurring at the intersection between rows in $[d]$ and columns in $X_\gamma$.
Using  Lemma \ref{lem:UjavgbigO} we obtain the following upper bound to $\HE{ \ \norm{U_j} \ }$.
\begin{equation}
 \label{eq:avgUj_upperbound}
 \HE{ \ \norm{U_j} \  } \; \le \; c' \left( \ 1 \ + \ \sqrt{ \ \gamma \ } \ \right) \; \sqrt{ \ \dfrac{d}{\bo} \ },
\end{equation} \noindent \noindent
where $c' \  \approx \  1.92$, and $\gamma \  = \ \frac{2(\so-1)}{d}$. For the case $\so=d$, the upper bound is simpler:
\begin{equation}
 \label{eq:avgUJ=j_upperbound_mequalsd}
 \HE{ \ \norm{U_j} \ } \ \le \  c' \  \sqrt{\dfrac{d-1}{n}},
\end{equation} 
where $c' \approx 3.85$. To simplify the presentation,  define 
\begin{equation}
 \label{eq:A}
A \coloneqq c' \left(1+\sqrt{\gamma}\right)\sqrt{\dfrac{d-1}{\bo}}.
 \end{equation} 
From  Lemma~\ref{lem:lip_normUj} it follows that the function  $U \rightarrow \norm{U_j}$ is $1$-Lipshitz on $\Unitary(\bo)$ with respect to the Hilbert-Schmidt metric. Therefore, the function satisfies the following concentration inequality (see Subsection  \ref{subsec:conc_measure})
\begin{equation}
 \label{eq:log-Sobolev1}
 \underset{U\sim \mu_\bo}{\pr}\left( \norm{U_j} \  \ge \  A \  + \  t  \right) \, \le \, \exp \left( \ - \ \dfrac{\bo \ t^2}{12} \  \right), \; \forall \; t \ge 0,
\end{equation} \noindent
where we have used the fact that
\begin{equation}
 \label{eq:log-Sobolev1_steps}
 \left\{ \ U \in \Unitary(\bo) \ \Bigg| \  \norm{U_j} \  \ge \  A \  + \  t  \right\} \, \subseteq \, \left\{ \ U \in \Unitary(\bo) \ \Bigg| \  \norm{U_j} \  \ge \  \HE{ \ \norm{U_j} \ } \  + \  t  \right\},  \; \forall \; t \ge 0 .
 \end{equation} \noindent
By defining
\begin{equation}
\label{eq:event_Ej}
\mathcal{E}_j \ := \ \left\{ \ U \in \Unitary(\bo) \ \Bigg| \ \norm{U_j}^2 \ \ge \ (A + t)^2 \ \right\}.
\end{equation} \noindent
we can rewrite the inequality \eqref{eq:log-Sobolev1} as
\begin{equation}
 \label{eq:log-Sobolev2}
 \underset{U\sim \mu_\bo}{\pr} \ \left( \  \mathcal{E}_j \  \right) \,  \le \,  \mathrm{exp} \left( - \frac{\bo \  t^2}{12} \right), \; \forall \; t \ge 0.
\end{equation} \noindent
Suppose $U$ be such that it satisfies: $\sum_{j=1}^\alpha \norm{U_j}^2 \ge  \alpha  (A+t)^2$. This implies that for at least one $j \in [\alpha]$, $U \in \mathcal{E}_j$. Using $\alpha \leq \frac{\bo}{\so-1}$, we obtain
\begin{align} 
 \label{eq:union_bound}
  \underset{U\sim \mu_\bo}{\pr}\left( \sum_{j=1}^\alpha \  \norm{U_j}^2 \  \ge \  \frac{\bo}{\so} \  (A+t)^2 \ \right)  \ \le \    \dfrac{\bo}{\so} \ \mathrm{exp}  \left( \  - \ \frac{ \bo \  t^2  }{12} \ \right), \; \forall \; t \ge 0,
\end{align} 
where we have used the union bound inequality on the event $\bigcup_{j=1}^\alpha \mathcal{E}_j$.
When $U$ satisfies the inequality $\sum_{j=1}^\alpha \  \norm{U_j}^2 \  \ge \  \frac{\bo}{\so} \  (A+t)^2$ then  using the fact that $ 1-t/A \ge (A+t)^{-2}$ when $0\le t/A  \le  \frac{\sqrt{5}-1}{2}$, we get that the success probability of our scheme is bounded by
 \begin{equation}
 \label{eq:event_manipulation1}
 \q^{(\so)}(\M^U) \  \le \ \dfrac{\so}{\bo \ A^2} \  \left( \  1 -  \dfrac{t}{A} \ \right), \; \mathrm{where} \;  0\le \frac{t}{A}  \le \frac{\sqrt{5}-1}{2}.
 \end{equation} \noindent
Finally, by taking $\epsilon :=  \frac{t}{A}$, and using equation \eqref{eq:A}, the event \eqref{eq:event_manipulation1} can be rewritten as
\begin{equation}
 \label{eq:eventq}
 \q^{(\so)}(\M^U) \  \le \  c \  \dfrac{\gamma}{\left( \ 1 + \sqrt{\gamma}  \ \right)^2} \  \left( \  1 - \epsilon \  \right), \; \mathrm{where} \;  0 < \epsilon < \frac{\sqrt{5}-1}{2},
\end{equation} \noindent
where $c = \frac{1}{2 \ c'^2} \  \approx 0.136$.
By plugging this into equation \eqref{eq:log-Sobolev3} we get the inequality \eqref{eq:conc_q_lower_technical}.
For the special case when $\so=d$, we follow the same reasoning as above, starting from inequality \eqref{eq:avgUJ=j_upperbound_mequalsd}, and then obtaining \eqref{eq:conc_q_lower_mequalsd}.
\end{proof}

\begin{thm}[Technical formulation of inequality~\eqref{eq:conc_q_upper_simpler1} from Theorem \ref{thm:conc_q_lower_upper}]
\label{thm:conc_q_upper_technical} 
Let $\bo \in \left\{ d,\ldots,d^2 \right\} $, $\so\leq d$. Let $\M^U$ denote a rank-one $\bo$-outcome Haar-random POVM on $\Complex^d$.
Let $q^{(\so)}(\M^U)$ be the maximal success probabilility of implementing $\M^U$ with postselection via convex combination of $\so$-outcome measurements. We then have 

\begin{equation}
 \label{eq:conc_q_upper_technical}
 \underset{U\sim \mu_\bo}{\pr} \left( q^{(\so)}(\M^U) \ \le \ \dfrac{2 \ \so \ \log \bo}{d} \ ( \ 1 \ + \ \epsilon \ ) \ \right)  \, \ge \, 1 \ - \ \dfrac{\e \ d}{\bo^{1+2\epsilon}}, \; \mathrm{where} \; \dfrac{1}{2 \log \bo} \ < \epsilon \  < \ \dfrac{\bo}{2 \ \log \bo } -1 \ .
\end{equation} \noindent
 
\end{thm}
\begin{rem}
 \label{rem5}
 Theorem \ref{thm:conc_q_upper_technical} is meaningful only for values of $d$, $\so$ and $\bo$ such that $2 \ \so \ \log \bo < d$. Moreover, inequality \eqref{eq:conc_q_upper_simpler1} is reproduced by setting $\epsilon =1$ in Eq.\eqref{eq:conc_q_upper_technical}.
\end{rem}

\begin{proof}
Let $\Sm$ be the set of all $\bo$-outcome POVMs simulable by quantum measurements with \emph{at most}  $\so$-outcomes. Let $\M$ be arbitrary $\bo$-outcome POVM on $\Complex^d$. To establish inequality \eqref{eq:conc_q_upper_technical} we shall use the following inequality between $q^{(\so)}(\M)$ and the robustness $R^{(\so)}$ (cf. Appendix~\ref{app:rela_q_R_V})
\begin{equation}
 \label{eq:Rb_q_relation}
 q^{(\so)}(\M) \le \dfrac{1}{R^{(\so)}(\M) +1}.
\end{equation} \noindent
 The robustness $R^{(\so)}(\M)$ has an operational interpretation: it can be expressed via the maximal relative advantage that $\M$ can offer over all over all possible POVMs in quantum state diecrimination  $\Sm$ (see Theorem 2, in \cite{OszmaniecBiswas2019}):
\begin{equation} 
 \label{eq:Rb_MED_gain}
 \Rb \ = \max_{\mathcal{E}} \dfrac{\mathrm{P}_{\mathrm{succ}} \left( \mathcal{E},\M \right)}{\underset{ \; \; \N \in \Sm}{\max}\,\mathrm{P}_{\mathrm{succ}} \left( \mathcal{E},\N \right)} \; \; - \; \;1,
\end{equation} \noindent
where $\mathcal{E} = \left\{ (q_i, \sigma_i) \right\}_{i=1}^{\bo}$ is an $\bo$-element ensemble of quantum states, and $\psucc{\M}$ ($\psucc{\N}$) is the success probability for the minimum error discrimination of the states with the POVM $\M$ (or $\N$ respectively).  For a given $\M$, we construct the following ensemble of states: 
\begin{equation}
 \label{eq:E_M}
 \mathcal{E}_{\M}\, := \, \left\{ (q_j, \sigma_j) \, \Big| \, q_j= \frac{1}{d}\tr M_i, \; \sigma_j = \frac{1}{\tr M_j} M_j  \ \right\}_{j=1}^{\bo}.
\end{equation} \noindent 
Note that the convexity of $\Sm$ implies that $\underset{ \; \; \N \in \Sm}{\max}\,\mathrm{P}_{\mathrm{succ}} \left( \mathcal{E},\N \right)$ is maximized on the extremal points of $\Sm$, which implies
\begin{equation}
\label{eq:ext_max}
\max_{X\subset[\bo], |X|=\so}\,\mathrm{P}_{\mathrm{succ}} \left( \mathcal{E},\N \right) \; \le  \; \max_{X\subset[\bo], |X|=\so} \;  \sum_{j \in X} \, q_j, 
\end{equation} \noindent
By using inequality \eqref{eq:Rb_q_relation}, and the fact that $\mathcal{E}_\M$ is a particular ensemble of quantum states (and that not-necessarily corresponding to the optimal value of the maximization in \eqref{eq:Rb_MED_gain}).   The obtain 
\begin{equation}
 \label{eq:q_MED_upperbound}
 q^{(\so)}(\M) \, \le \frac{\max_{X\subset[\bo], |X|=\so} \sum_{j \in X} \, \tr M_j }{\sum_{j=1}^{\bo} \,  \tr \, M_j^2}.
\end{equation} \noindent
Now let $\M$ be a rank-one $\bo$-outcome POVM, whose $j$-th effect takes the form $M_j = w_j \kb{\psi_j}{\psi_j}$, where $\bk{\psi_j}{\psi_j}=1$, for all $j$ and $w_j = \tr M_j$. For this choice of $\M$, we get 
\begin{equation}
 \label{eq:q_MED_upperbound_rank_one}
 q^{(\so)}(\M) \, \le \, \dfrac{\max_{X\subset[\bo], |X|=\so}\,  \sum_{j \in X} \, w_j}{\sum_{j=1}^{\bo} \,  w_j^2}.
\end{equation} \noindent
If $\M=\M^U$ then we have $w_j = \sum_{i=1}^d |U_{ij}|^2$, and $\sqrt{w_j} \ \bk{e_i}{\psi_i}= U_{ij}$, for $i \in [d]$, and $j \in [\bo]$. From Lemma \ref{lem:Dr}, it is seen that the minimum value of $\sum_{j=1}^{\bo} \,  w_j^2$ for any $\bo \times \bo$ unitary $U$ is $\frac{d^2}{n}$. Thus we get 
\begin{equation}
 \label{eq:q_MED_upperbound_rank_one_2}
 q^{(\so)}(\M) \, \le \, \frac{n}{d^2}\,\left( \max_{X\subset[\bo], |X|=\so}\,  \sum_{j \in X} \, w_j\right).
\end{equation} \noindent
When $U$ is distributed according to the Haar measure, then we can use inequality \eqref{eq:q_MED_upperbound_rank_one_2} from Lemma \ref{lem:Nr}, which proves the theorem.

\end{proof}

\section{Effects of depolarizing noise on the implementation of POVMs}
\label{sec:dop}

In this section we provide detailed description of some concepts which appeared in the "Noise analysis" section of the main text.
This includes description of how completely depolarizing noise on the level of quantum circuits propagates into POVMs implemented via two schemes -- Naimark's dilation (Section~\ref{sec:naimark_noise}) and the scheme introduced in this work.
In Section~\ref{sec:proposition_proof} we give a proof of Proposition~\ref{thm:dop_lower} from the main text.
We finish this section by providing some details and brief motivation behind the used noise model.

\subsection{Depolarizing noise in Naimark's dilation}
\label{sec:naimark_noise}

In the main text, we described how the depolarizing noise acts on the quantum measurements. 
However, as already noted in Subsection~\ref{subsec:Haar_random_POVM}, in actual implementations to perform change of basis required by Naimark's dilation, one usually implements (adjoint) unitaries acting on the states, i.e.,
\begin{align}\label{eq:born_rule_app}
    \tr\rbracket{\rho \otimes \ketbra{0}{0}UP_iU^{\dag}} = \tr\rbracket{U^{\dag}\rho \otimes \ketbra{0}{0} UP_i} \ ,
\end{align}
where $\cbracket{P_i}$ is a computational basis measurement
on extended Hilbert space, $\rho$ is a state we want to perform a POVM on, while $\ketbra{0}{0}$ and $U$ are an ancilla and unitary required by Naimark's dilation (we note that ancilla can be taken as $\ketbra{0}{0}$ without loss of generality).
To be explicit, in such implementation, we effectively implement on the system $\rho$ the quantum measurement with effects given by $M_i^{U} = \tr_B \rbracket{\unity \otimes \kb{0}{0} \ U P_i U^{\dag} }$ (where $B$ denotes second, ancillary system) with superscript $U$ indicating that the POVM is associated with quantum circuit $U$.

From the above it follows that if the change-of-basis unitaries are affected by noise, it will impact the implementation of a target POVM.
If a completely depolarizing noise with visibility $\eta$ acts on the (rotated) quantum state $\sigma \coloneqq U^{\dag}\rho \otimes \ketbra{0}{0} U$, it changes it as
\begin{align}
    {\sigma} \rightarrow \eta {\sigma} + \rbracket{1-\eta} \frac{\unity_n}{n} \ ,
\end{align}
where we use $n$ as label for dimension of the total system which is equal to the number of outcomes of the target POVM. Putting this noisy state into Eq.~\eqref{eq:born_rule_app} gives
\begin{align}
    \tr\rbracket{\rbracket{\eta {\sigma} + \rbracket{1-\eta}\frac{\unity}{n}}P_i} =     \tr\rbracket{{\sigma}\rbracket{\eta P_i +\rbracket{1-\eta}\frac{\unity}{n}}} \ ,
\end{align}
where we used the fact that $\tr\rbracket{P_i}=1$ for each rank-1 projector $P_i$.
Hence we see that performing perfect measurement $\cbracket{P_i}$ on noisy quantum sate $\eta \sigma +\rbracket{1-\eta} \frac{\unity}{n}$ is experimentally equivalent to having a perfect state $\sigma$ and performing noisy quantum measurement with effects distorted as $P_i \rightarrow \eta P_i +\rbracket{1-\eta} \frac{\unity}{n}$. Now since effects of (ideal) target POVM $\M$ are given (via Naimark's dilation) by $M^{U}_i = \tr_B \rbracket{\unity \otimes \kb{0}{0} \ U P_i U^{\dag} }$, we get that after the action of the noise channel, the effects of our target POVM are distorted as: $
    M_i^{U} \rightarrow \eta M^{U}_i + \rbracket{1-\eta}  \frac{\unity}{n}$, where $\unity$ is the identity operator on $\Complex^d$. 
This motivates defining depolarized version $\M^{\eta}$ of the measurement $\M$ with effects given by
\begin{equation}
\label{eq:depolarising_noise}
M_i^{U,\eta} \coloneqq \eta M_i^{U} + (1-\eta) \frac{\unity}{n} .
\end{equation}


\subsection{Proof of Proposition \ref{thm:dop_lower} in the main text}
\label{sec:proposition_proof}

We start by stating the formal definition of Total-Variation Distance (TVD) followed by reiterated Proposition~\ref{thm:dop_lower} from the main text.
\begin{definition}[Total variation distance]
\label{def:tv}
Let $\p{p}$ and $\p{q}$ be two $\bo$-probabilities, with $i$-th outcomes $p_i$ and $q_i$ respectively. Then the total variation distance between $\p{p}$ and $\p{q}$ is defined as 
\begin{equation}
 \label{eq:tv_pq}
 \tv \left(\p{p} , \p{q} \right)\;  \coloneqq \; \frac{1}{2} \, \sum_{j=1}^{\bo} \, \left| \ p_i - q_i \ \right|.
\end{equation} \noindent
\end{definition}

\begin{thm}(Proposition \ref{thm:dop_lower} of the main text) 
Let $\M^{U,\eta}$ be the noisy implementation of Haar-random POVM $\M^{U}$ associated with unitary $U$ (see Subsection~\ref{subsec:Haar_random_POVM}), with effects given by
\begin{equation}
 \label{eq:nsyM2}
M^{U,\eta}_i \ := \ \eta M^{U}_i \ + \ (1-\eta) \frac{\unity}{\bo},
\end{equation} \noindent
where $\eta \in [0,1]$
Then we have
\begin{equation}
 \label{eq:dop_m_nsym}
 \HE{ \ \max_{\rho} \  \tv \left(\p{p}\rbracket{\M^{U}|\rho} , \p{p}\rbracket{\M^{U,\eta}|\rho} \right)  \ } \;  \ge \; \left(1-\eta \right) c_n \ ,
\end{equation} \noindent
where $\pv{\M^{U}}{\rho}$ ($\pv{\M^{U,\eta}}{\rho}$) is a probability distribution obtained via Born's rule when measurement $\M^{U}$ ($\M^{U,\eta}$) is performed on the state $\rho$, and
\begin{align}
    c_n = \rbracket{1-\frac{1}{n}}^{n} \approx \frac{1}{e} \ .
\end{align}
\end{thm}
\begin{proof}
The completely depolarising noise model specified in equation \eqref{eq:nsyM2},
allows to relate the probabilities $\pv{\M^{U}}{\rho}$ and $\pv{\M^{U,\eta}}{\rho}$ for any quantum state $\rho$ as follows
\begin{align}
 \label{eq:p_M_nsyM_rela}
 \pv{\M^{U,\eta}}{\rho} \, = \, \eta \  \pv{\M^{U,\eta}}{\rho} \, + \,  (1-\eta) \ \mathbf{p}^{\text{n}} \ ,
\end{align}
where $\mathbf{p}^{\text{n}}$ is a \textit{uniform probability distribution} over $n$ outcomes.
After basic manipulations, this gives that the total variation distance between $\pv{\M^{U}}{\rho}$ and $\pv{\M^{U,\eta}}{\rho}$ is 
\begin{equation}
 \label{eq:vt_rho_M_nsyM}
 \tv \left(\pv{\M^{U}}{\rho},\pv{\M^{U,\eta}}{\rho}\right) \, = \, \left(1-\eta\right) \ \tv \left( \pv{\M^{U}}{\rho}, \mathbf{p}^{\text{n}}\right) \ ,
\end{equation} \noindent
where, explicitly, 
\begin{equation}
 \label{eq:tv_explicit_rho}
 \tv \left(\pv{\M^{U}}{\rho},\mathbf{p}^{\text{n}}\right) = \; \dfrac{1}{2} \; \sum_{j=1}^{\bo} \; \left| \ \tr\rbracket{\rho M^{U}_j} -\frac{1}{n} \ \right|.
 \end{equation} \noindent
 Hence it turns out that total-variation distance between distributions generated by ideal and completely-depolarized version of $\M$ is proportional to a distance between the original distribution $\pv{\M^{U}}{\rho}$ and completely random distribution $\mathbf{p}^{\text{n}}$.
This observation will greatly simplify further considerations.
Namely, recall that we are interested in bounding (expected value of) the worst-case (over quantum states) error in TVD of LHS of Eq.~\eqref{eq:tv_explicit_rho}.
We now see that it is equivalent to providing bound for the RHS of Eq.\eqref{eq:tv_explicit_rho}, which is easier to handle.

To start, recall that the matrix elements of $M_j$ are directly related to the matrix elements of Naimark's unitary $U$ via $\left( \ M^{U}_{i} \  \right)_{j} = |U_{ij}|^2$ (see Subsection~\ref{subsec:Haar_random_POVM}).
Now, since maximal value of any function is lower-bounded by any of the function's value, by choosing input state $\rho = \kb{e_i}{e_i}$ for some $i \in [d]$, and inserting it into equation \eqref{eq:tv_explicit_rho}, we get that
\begin{align}
\max_{\rho} \  \tv \left(\p{p}\rbracket{\M^{U}|\rho} , \mathbf{p}^{\text{n}}\right)\; \ge \; \dfrac{1}{2} \; \sum_{j=1}^{\bo} \, \left| \  \bra{e_i} M^{U}_j \ket{e_i} - \frac{1}{\bo} \ \right| \; = \; \frac{1}{2} \, \  \sum_{j=1}^{\bo} \  \left| \ \left|U_{ij} \right|^2 - \frac{1}{\bo} \right|.
\end{align}
Thus the expected value is lower bounded as follows.
 \begin{align}
 \HE{ \max_{\rho} \  \tv \left(\p{p}\rbracket{\M^{U}|\rho} , \mathbf{p}^{\text{n}}\right)\ }  \,
  \ge  \, \HE{ \ 
\frac{1}{2} \ \sum_{j=1}^{\bo} \ \left| \ \left|U_{ij} \right|^2 \ - \ \frac{1}{\bo} \ \right|
}.
\end{align} 
The permutational invariance of the Haar measure implies that
\begin{equation}
 \label{eq:thm4_1}
  \HE{ \ 
\frac{1}{2} \ \sum_{j=1}^{\bo} \ \left| \ \left|U_{ij} \right|^2 \ - \ \frac{1}{\bo} \ \right|
} \ = \frac{\bo}{2} \ \HE{ \  \left| \ \left|U_{ij} \right|^2 \ - \ \frac{1}{\bo} \ \right| \ }.
\end{equation} \noindent

Now we note that since $U$ is Haar-random, the $|U_{ij}|^2$ has the same distribution as $x$ from 
Eq.\eqref{eq:px}, i.e., $p(x) = \left( \bo - 1 \right) \left( 1-x\right)^{\bo-2}$.
This allows to perform integration as
\begin{equation}
 \label{eq:thm4_2}
 \HE{ \  \left| \ \left|U_{ij} \right|^2 \ - \ \frac{1}{\bo} \ \right| \ } \ = \ \int_0^1 \ dx \ \left| \ x - \frac{1}{\bo} \right| \ p(x) \  = \ \frac{2}{\bo} \  \left( \ 1- \frac{1}{\bo} \ \right)^{\bo}.
\end{equation} \noindent
Thus we get 
\begin{equation}
 \label{eq:dop_2}
 \HE{ \max_{\rho} \  \tv \left(\p{p}\rbracket{\M^{U}|\rho} , \mathbf{p}^{\text{n}}\right)\ } \, \ge  \, \left( \ 1- \frac{1}{\bo} \ \right)^{\bo} \approx 1/e \ .
 \end{equation} \noindent
Combining inequality \eqref{eq:dop_2} with equation \eqref{eq:vt_rho_M_nsyM} proves the theorem.  
\end{proof}

\subsection{Depolarizing noise in implementation with post-selection}

In this part we study how global depolarizing noise acting affects the quiality of our POVM implementation scheme involving postselection. 
Recall that our scheme implements a measurement
\begin{align}
    \mathbf{N} = \rbracket{q M_1, \dots, q M_{n}, \rbracket{1-q}\unity} \ ,
\end{align}
where $\M$ is a target POVM (which we assume consist of rank one effects) and  $q$ is a success probability of the implementation. The above measurement is realized as a convex mixture of $\so$-outcome measurements (for simplicity we assume here that $\so-1$ divides $n$) as
\begin{align}
    \mathbf{N} =  \sum_{\gamma}p_{\gamma} \mathbf{N}^{\gamma} \ ,
\end{align}
where each $\N^{\gamma}$ has $n+1$ formal outcomes, such that
\begin{align}\label{eq:postselection_effects}
    N^{\gamma}_i = \begin{cases} 
     \lambda_{\gamma} M_i \ & \text{if}\ i \in \gamma \ ,\\
     \unity - \lambda_{\gamma} \sum_{i\in X_{\gamma}} M_i & \text{if}\ i=\bo+1\ , \\
    0\ & \text{if} \ i \in \sbracket{n} \backslash X_{\gamma}\ , \\
    \end{cases}
\end{align}
where $X_{\gamma}$ is subset of $|X_{\gamma}| \leq m-1$ outcomes and probability distribution $\lbrace p^\gamma \rbrace$ is defined by
\begin{align}
    p_{\gamma} = \frac{q}{\lambda_{\gamma}}\ ,\  \lambda_{\gamma} = ||\sum_{i \in X_{\gamma}} M_i||^{-1}\ , \ 
    q = \rbracket{\sum_{\gamma} \frac{1}{\lambda_{\gamma}}}^{-1} = \rbracket{\sum_{\gamma=1}^{\alpha}\ ||\sum_{i\in X_{\gamma}}M_i||}^{-1}  \ .
\end{align}
Each of the measurements  $\cbracket{\mathbf{N}^{\gamma}}$ is implemented via Naimark's dilation theorem (i.e projective POVM on extended Hilbert space). As explained in the main text, if the target POVM $\M$ is rank one, and $\so\leq d$ then POVMs $\N^\gamma$ can be implemented using Hilbert space of dimenstion $\so-1+d \leq 2d =: d_{tot}$.
Now, due to the noise, the effects of the implemented POVM are distorted as
\begin{align}\label{eq:distorted_post}
    N^{\gamma}_{i} \rightarrow \eta N^{\gamma}_i + \rbracket{1-\eta} \frac{\unity}{d_{tot}} \ ,\ \text{for}\ i\notin \gamma \cup \lbrace \bo+1 \rbrace\ .  
    \end{align}
Therefore, in the presence of the assumed noise model our protocol effectively implements a POVM $ \N^{\eta} \coloneqq \sum_{\gamma}p_{\gamma} \N^{\eta,\gamma}$ , where by  $\N^{\gamma,\eta}$ we denoted indicate noisy veriat of POVM $\N^\gamma$, with effects given in \eqref{eq:distorted_post}. 

We are interested in bounding the distance between target distribution $\cbracket{p\rbracket{i|\M,\rho}}_{i=1}^{n}$ and the post-selected distribution from noisy POVM $\N^{\eta}$, i.e., the distance, 
\begin{align}\label{eq:postselected_distance}
  \mathrm{d_{TV}}\rbracket{\pv{\M}{\rho},\mathbf{p}^{\mathrm{noisy}}_{\mathrm{post}}(\M|\rho)} = \frac{1}{2}\sum_{i=1}^{n}| p\rbracket{i|\M,\rho}-\frac{p\rbracket{i|\N^{\eta},\rho}}{p\rbracket{i< n+1|\N^{\eta},\rho}}|  
\end{align}
where we have used  
\begin{equation}
p^{\mathrm{noisy}}_{\mathrm{post}}(i|\M,\rho)= \frac{p\rbracket{i|\N^{\eta},\rho}}{p\rbracket{i< n+1|\N^{\eta},\rho}} \ .
\end{equation}
Let $\gamma(i)$ to denote the label of the subset of outcomes to which $i$ belongs. Consequently we have 
\begin{align}
    p_{\gamma(i)} = p_{\gamma} \ \text{for all }i\in X_{\gamma} \ .
\end{align}
We note that for $i\in \gamma$ we have
\begin{align}\label{eq:probability_distortion_post}
p\rbracket{i|\N^{\eta},\rho}=  p_{\gamma(i)}\  p\rbracket{i|\N^{\gamma,\eta},\rho} = p_{\gamma(i)}\ \rbracket{\eta \   p\rbracket{i|\N^{\gamma},\rho} + \rbracket{1-\eta} \frac{1}{d_{tot}}}  = \eta\    q\  p\rbracket{i|\M,\rho} +\rbracket{1-\eta} \frac{p_{\gamma(i)}}{d_{tot}} \ , 
\end{align}
where we used the fact that $p_{\gamma(i)}p\rbracket{i|\N^{\eta},\rho} = q\  p\rbracket{i|\M,\rho}$.
To understand the behaviour of Eq.~\eqref{eq:postselected_distance} we need to calculate how probability of postselection  changes due to the noise. Using the fact that subsets $\\gamma$ are disjoint end employing \eqref{eq:probability_distortion_post} we obtain 
\begin{align}\label{eq:postselection_probability_noisy}
    p\rbracket{i< n+1|\N^{\eta},\rho} = \sum_{\gamma=1}^{\alpha} p_{\gamma} \sum_{i\in X_{\gamma}} p\rbracket{i|\N^{\eta,\gamma},\rho} = \eta\ q +\rbracket{1-\eta} \frac{\left<|X_{\gamma}|\right>}{d_{tot}} \ ,
\end{align}
where we defined $\left<|X_{\gamma}|\right> \coloneqq \sum_{\gamma=1}^\alpha p_{\gamma} |X_{\gamma}|$ and used the fact that $\sum_{i} p\rbracket{i|\M,\rho}=1$.

Now we rewrite the Eq.~\eqref{eq:postselected_distance} as
\begin{align}\label{eq:postselected_distance_funny}
      \frac{1}{2\ p\rbracket{i< n+1|\N^{\eta},\rho}} \sum_{i=1}^{n}| p\rbracket{i< n+1|\N^{\eta},\rho}\ p\rbracket{i|\M,\rho}-p\rbracket{i|\N^{\eta},\rho}| \ ,  
\end{align}
We calculate each of the summands explicitly using  Eq.~\eqref{eq:probability_distortion_post} and Eq.~\eqref{eq:postselection_probability_noisy} and obtain  
\begin{align}
    |p\rbracket{i< n+1|\N^{\eta},\rho}\ p\rbracket{i|\M,\rho}-p\rbracket{i|\N^{\eta},\rho}| = \frac{\rbracket{1-\eta}}{d_{tot}} \  |\ p\rbracket{i|\M,\rho}\ \left<|X_{\gamma}|\right> - p_{\gamma(i)} \ | \ ,
\end{align}
Using the bound $|a-b|\leq |a|+|b|$ and summing over $i$ we obtain
\begin{align}\label{eq:inequality_sum_post}
    \frac{\rbracket{1-\eta}}{d_{tot}} \  \sum_{i=1}^{n} |\ p\rbracket{i|\M,\rho}\ \left<|X_{\gamma}|\right> - p_{\gamma(i)} \ | \ & \leq       \frac{\rbracket{1-\eta}}{d_{tot}} \  \sum_{i=1}^{n} \ \rbracket{\  p\rbracket{i|\M,\rho}\ \left<|X_{\gamma}|\right> + p_{\gamma(i)}  }  =     \frac{2\rbracket{1-\eta} \left<|X_{\gamma}|\right>}{d_{tot} }  \ . 
\end{align}
Consider our scheme for the special choice $n=d^2$ and $\so = d+1$, hence $d_{tot} = 2d$ and  $|X_{\gamma}| = d$ for all $\gamma$. This gives  $\frac{2\left<|X_{\gamma}|\right>}{d_{tot}} = 1$. 
Combining this with the inequality in Eq.~\eqref{eq:inequality_sum_post} and the Eq.~\eqref{eq:postselected_distance_funny}  yields that for our scheme we have
\begin{align}\label{eq:noisy_post_bound}
\mathrm{d_{TV}}\rbracket{\pv{\M}{\rho},\mathbf{p}^{\mathrm{noisy}}_{\mathrm{post}}(\M|\rho)} \leq \frac{1}{2}\  \frac{\rbracket{1-\eta}}{\ p\rbracket{i< n+1|\N^{\eta},\rho}} = \frac{1}{2} \  \frac{\rbracket{1-\eta}}{\eta \ q + \rbracket{1-\eta} \frac{1}{2}}\leq (1-\eta)\max\lbrace\frac{1}{2q},1 \rbrace  \ . 
\end{align}
For Haar-random rank-one POVMs we have $\q(\M^U)>c$ (see Theorem \ref{thm:conc_q_lower_upper}), where $c$
 is an absolute constant. Combining this with the fact that for generic unitaries on $2N$ qubits we have $\eta^{\mathrm{post}}= \exp(-\Theta(4^N)) $ we obtain the assertion made in the main text, i.e, that for typical Haar-random $d^2$-outcome POVMs $\M^U$ we have
 \begin{equation}
    \mathrm{d_{TV}}\rbracket{\pv{\M^U}{\rho},\mathbf{p}^{\mathrm{noise}}_{\mathrm{post}}(\M^U|\rho)}\leq C (1-\exp\rbracket{-\Theta\rbracket{4^N}})\ .
\end{equation}

\subsection{Noise model details}

In the main text and in previous subsections, we adopted a very simple noise model parametrized by only single number -- visibility $\eta$.
The main motivation for that choice was the fact that since we consider mostly generic Haar-random POVMs, the circuits which implement them can be considered random, and that such model was considered in Google's recent demonstration of computational advantage (which used random circuits) \cite{Google2019}. 
In Ref.~\cite{Google2019}, authors consider $\eta$ of the following form \cite{Boixo2018}
\begin{equation}
\label{eq:eta}
    \eta = \exp\left(-r_1g_1 -r_2g_2-N(r_p + r_m)\right),
\end{equation}
where $r_1$, $r_2$ are respectively the error rates for single and two-qubit gates, $g_1$, $g_2$ are number of single-qubit and two-qubit gates, $N$ is the total number of qubits in the circuit, and $r_p$ and $r_m$ are SPAM (state preparation and measurement) errors.
As indicated in the main text, since generic circuits require number of two-qubit gates scaling exponentially with the system size, we considered faulty two-qubit gates as the main error source.
We note, however, that in the above model it is in fact assumed that readout noise can be effectively treated as uncorrelated and identical.
In presence of the measurement noise cross-talk, the more realistic noise model should be considered (see, for example, recent works \cite{Bravyi2020,Maciejewski2021}).






\section{Numerical results} \label{sec:numerics}

\subsection{IC and SIC POVMs}

\subsubsection{Informationally complete measurements covariant with respect to $\mathbb{Z}_d \times\mathbb{Z}_d$}
To explain how we construct informationally complete (IC) measurements, let us first recall that a POVM is called covariant with respect to a group, if all of the measurement operators can be obtained from some \textit{fiducial vector} by the action of that group.
Hence if one has a way of constructing that fiducial vector and the unitary representation of chosen group, one can easily generate all of the effects of covariant measurement.
In this work we use the explicit construction from Ref.~\cite{DAriano2004} which shows how to obtain fiducial vector for the POVM covariant with respect to $\mathbb{Z}_d \times\mathbb{Z}_d$ (which can be thought of as finite-dimensional analogue of Weyl-Heisenberg group), where $d$ is the dimension of the system.
Such POVM has $d^2$ rank-1 effects and is shown to be informationally-complete \cite{DAriano2004}.
A fiducial vector is constructed as
\begin{align}
  \ket{\psi_\alpha} =   \sqrt{\frac{1-|\alpha|^2}{1-|\alpha|^{2d}}} \ \sum_{i=0}^{d-1}\alpha^i \ket{i},
\end{align}
where $\alpha$ is a parameter characterizing the POVM and has to fulfill condition $0<|\alpha|<1$
Now vectors defining other effects of that POVM are obtained as
\begin{align}
    \ket{\psi_{m,n}}=\frac{1}{\sqrt{d}}U_{m,n}\ket{\psi_\alpha} \ ,
\end{align}
where $U_{m,n}$ is a (projective) unitary representation of $\mathbb{Z}_d \times\mathbb{Z}_d$ given by

\begin{align}
    U_{m,n} = \sum_{k=0}^{d-1}\ \exp\left(\frac{2\pi i}{d}km\right)\ket{k}\bra{k\oplus n} \ ,
\end{align}
with $m,n\in \left[0,d-1\right]$ and $\oplus$ is addition modulo $d$.
See Ref.~\cite{DAriano2004} for more details. In our simulations we arbitrarily choose the free parameter to be $   \alpha = \frac{1}{2}\left(1+i\right)$. We note that we checked a few other instances of this parameter and we did not observe quantitative differences in the probability of success of POVMs simulation using our scheme.

\subsubsection{Symmetric Informationally Complete measurements}
The measurement is called symmetric if its effects have equal pairwise Hilbert-Schmidt scalar products.
The search for symmetric and informationally complete (SIC) measurements is an active area of research \cite{Fuchs2017SIC} and even existence of SICs in arbitrary dimension $d$ is an open problem.
To date, SIC POVMs have been found either numerically or analytically for a restricted collection of dimensions 
\cite{Scott2010symmetric,scott2017sics,Grassl2017fibonacci}.
SIC POVMs are, similarly to IC, represented by a single fiducial vector and we generate other measurement operators from that vector by the action of $\mathbb{Z}^d \times\mathbb{Z}_d$ group (we note that all SIC POVMs found to date are covariant with respect to some group, and the most of them covariant to $\mathbb{Z}^d \times\mathbb{Z}_d$ group). 

In this work, the POVMs in dimensions $d\in \left[2,100\right]$ have been downloaded from database \cite{SICPOVM_catalogue} maintained by Christopher A. Fuchs, Michael C. Hoang, and Blake C. Stacey. 
The POVMs for dimensions
\begin{align}
    \left[100,193\right] \cup \cbracket{194,195,201,204,224,228,255,259,288,292,323,327,364,399,403,489,528,725,844,1155,1299}
\end{align}
were provided by Markus Grassl in private correspondence.


\subsection{Haar-random POVMs}

In this work, we are interested in generating Haar random $d$-dimensional POVMs with $d^2$ outcomes.
A straightforward method to do so would be to generate Haar-random $d^2 \times d^2$ unitary matrix and take its $d^2 \times d$ submatrix as defining such POVM.
However, generation of random matrices quickly becomes unfeasible -- due to large amount of memory required, we were not able to generate such matrices for high $d$.
As a workaround, instead of generating random $d^2 \times d^2$  unitary matrices, we generated random $d^2 \times d$ isometries.
To do so, we implemented the following algorithm.
\begin{enumerate}
\item Generate $d$ iid random complex Gaussian vectors of size $d^2$ -- call them $\cbracket{\mathbf{v}_i}_{i=1}^d$.
\item Construct a Gramian matrix $G$  of those vectors as
\begin{align}
    G_{ij}  = \braket{\mathbf{v}_i}{\mathbf{v}_j} \ .
\end{align}
\item Perform LDL decomposition of the Gramian matrix as
\begin{align}
    G = L\sqrt{D}\sqrt{D}L^{\ast} \ ,
\end{align}
where $L$ is lower-triangular and $D$ diagonal.
\item Define $R = \rbracket{\sqrt{D}L^{\ast}}^{-1}$ and construct new set of vectors as
\begin{align}
    \mathbf{e}_k= \sum_{i} R_{i,k} \mathbf{v}_i \ . 
\end{align}
It follows that $\cbracket{\mathbf{e}_k}_{k=1}^{d}$ forms an orthonormal set of $d^2$-dimensional random vectors.
Hence those vectors can be used to construct a $d^2 \times d$ isometry.
\item 
To construct a POVM one simply looks at rows of this isometry as a set of $d^2$ vectors of dimension $d$.
Since the matrix is an isometry, it follows that those rows define rank-1 effects of $d^2$-outcome random POVM.
\end{enumerate}


\end{document}